\documentclass[submission,copyright,creativecommons]{eptcs}
\providecommand{\event}{GCM 2020} % Name of the event you are submitting to

\title{Efficient Computation of Graph Overlaps\\ for Rule Composition: Theory and Z3 Prototyping}

\author{
Nicolas Behr\thanks{Corresponding author; thanks to the Bettencourt Schueller Foundation long term partnership, this work was partly supported by the CRI Research Fellowship to Nicolas Behr.}
\institute{Center for Research and Interdisciplinarity\\
Universit\'{e} de Paris, INSERM U1284\\
8-10 Rue Charles V, 75004 Paris, France}
\email{nicolas.behr@cri-paris.org}
\and
Reiko Heckel
\institute{School of Informatics\\ 
University of Leicester\\
LE1 7RH Leicester, UK}
\email{rh122@leicester.ac.uk}
\and
Maryam Ghaffari Saadat 
\institute{School of Informatics\\ 
University of Leicester\\
LE1 7RH Leicester, UK}
\email{mgs17@leicester.ac.uk}
}

\def\titlerunning{Computing Graph Overlaps for Rule Composition}
\def\authorrunning{N. Behr, R. Heckel \& M. Ghaffari Saadat}

\usepackage{graphicx}
\usepackage[table]{xcolor}
\usepackage{amsmath,amssymb}
\usepackage{mathtools}
\usepackage{longtable}
\usepackage{booktabs}
\usepackage{wrapfig}
\usepackage[binary-units=true]{siunitx}
\usepackage[binary-units=true]{siunitx}

\usepackage{array}

\usepackage[ddmmyyyy,hhmmss]{datetime}

\usepackage{amsmath, amssymb,amsthm}

\theoremstyle{plain}
\newtheorem{corollary}{Corollary}

\newtheorem{theorem}{Theorem}
\theoremstyle{definition}
\newtheorem{definition}{Definition}
\newtheorem{example}{Example}
\newtheorem{remark}{Remark}

\newtheorem{assumption}{Assumption}{\bfseries}{\itshape}

\makeatletter
\def\bstr{b}
\def\bfstr{bf}
\def\cstr{c}
\def\fstr{f}
\def\strLst{A,B,C,D,d,E,F,G,H,I,J,K,L,M,N,O,P,Q,R,S,T,U,V,W,X,Y,Z}

\newcommand{\MkB}[1]{\expandafter\def\csname\bstr#1\endcsname{\mathbb{#1}}}
  \@for\i:=\strLst\do{%
    \expandafter\MkB \i     }
    
\newcommand{\MkBF}[1]{\expandafter\def\csname\bfstr#1\endcsname{\mathbf{#1}}}
  \@for\i:=\strLst\do{%
    \expandafter\MkBF \i     }

\newcommand{\MkCal}[1]{\expandafter\def\csname\cstr#1\endcsname{\mathcal{#1}}}
  \@for\i:=\strLst\do{%
    \expandafter\MkCal \i     }
    
\newcommand{\MkFrak}[1]{\expandafter\def\csname\fstr#1\endcsname{\mathfrak{#1}}}
  \@for\i:=\strLst\do{%
    \expandafter\MkFrak \i     }
    
\makeatother

\newcommand{\Lin}[1]{\mathop{\mathsf{Lin}}(#1)}

\newcommand{\LinEq}[1]{\overline{\mathsf{Lin}}(#1)_{\sim}}
\newcommand{\LinAc}[1]{\overline{\mathsf{Lin}}(#1)}

\newcommand{\cond}[1]{\mathsf{cond}(#1)}

\newcommand{\ac}[1]{\mathsf{#1}} %-- \mathsf version
\newcommand{\Shift}{\mathsf{Shift}}

\newcommand{\Trans}{\mathsf{Trans}}

\newcommand{\OMatch}[2]{\mathsf{M}^{{\text{\tiny sq}}}_{#1}(#2)}

\newcommand{\RMatch}[2]{\mathsf{M}^{{\text{\tiny sq}}}_{#1}(#2)}

\newcommand{\pB}[1]{\mathsf{PB}(#1)}
\newcommand{\pO}[1]{\mathsf{PO}(#1)}

\newcommand{\cSq}[1]{\square_{#1}}

\newcommand{\mono}[1]{\mathsf{mono}(#1)}

\newcommand{\iso}[1]{\mathsf{iso}(#1)}

\newcommand{\obj}[1]{\mathsf{obj}(#1)}

\newcommand{\mIO}{\mathop{\varnothing}}

\newcommand{\sqComp}[3]{#1 \stackrel{#2}{\sphericalangle} #3}

\newcommand{\inputtikz}[1]{%
	\vcenter{\hbox{\includegraphics{diagrams/#1.pdf}}}%
}

\newcommand{\inputtikzB}[2]{%
 \raisebox{#1cm}{\includegraphics{diagrams/#2.pdf}}%
}

\usepackage[draft,cachedir=./mintedCache]{minted}

\usepackage{inconsolata}

\setminted[python]{
	linenos=true,
	frame=none,
	xleftmargin=2.5em,
	fontsize=\footnotesize
}

\setmintedinline[python]{
	fontsize=\normalsize
}

\newcommand\pythoninline[1]{\mintinline{python}{#1}}
\colorlet{h1color}{blue!70!black} % highlight color 1
\colorlet{h2color}{orange!90!black} % highlight color 1
\colorlet{h3color}{blue!40!white} % highlight color 1
\colorlet{h4color}{green!40!black} % highlight color 1
\newcolumntype{L}{>{$}l<{$}} % math-mode version of "l" column type

\usepackage[strings]{underscore}
\usepackage[all]{hypcap}

\usepackage{hyperref}

\begin{document}
\maketitle              % typeset the header of the contribution
\begin{abstract}
Graph transformation theory relies upon the composition of rules to express  the  effects of sequences of rules. In practice, graphs are often subject to  constraints, ruling out many candidates for composed rules. Focusing on the case of sesqui-pushout (SqPO) semantics, we develop a number of alternative strategies for computing compositions, each theoretically and with an implementation via the Python API of the Z3 theorem prover. The strategies comprise a straightforward generate-and-test strategy based on forbidden graph patterns, a variant with a more implicit logical encoding of the negative constraints, and a modular strategy, where the patterns are decomposed as forbidden relation patterns. For a toy model of polymer formation in organic chemistry, we compare the performance of the three strategies in terms of execution times and memory consumption. 
\end{abstract}

%%%%%%%%%%%%%%%%%%%%%%%%%%%%%%%%%%%%%%%%%%%%%%%%%%
\section{Introduction}
%%%%%%%%%%%%%%%%%%%%%%%%%%%%%%%%%%%%%%%%%%%%%%%%%%

When applying graph transformation rules sequentially or in parallel, we are often interested in capturing their combined preconditions and effects in a single composed rule, e.g. to compare the effects of derivations during confluence analysis~\cite{DBLP:journals/fuin/EhrigGHLO12}, to show the preservation of behaviour under different composition operations on  rules~\cite{DBLP:journals/entcs/LevendovszkyPE07,DBLP:conf/gg/RangelLKEB08}, or to define morphisms between systems that map individual to composed rules~\cite{DBLP:journals/mscs/HeckelCEL96}. %
Usually, a notion of rule composition, such as \emph{concurrent productions} in the algebraic approaches, relies on a specification of the relation or overlap of the given rules followed by a construction of the combined rule. A particular application scenario for such types of computations are stochastic graph rewrite systems (SRS)~\cite{bp2018,nbSqPO2019,bdg2019,bp2020}, a categorical notion of stochastic rewrite systems such as \texttt{Kappa}~\cite{danos2004formal,DanosFFHK08,Boutillier:2018aa}, using an algebraic structure on rewrite steps to support the analysis of biological systems. Instead of stochastic simulation or model checking with their well-known scalability issues, SRS aim at the derivation of a system of evolution equations which, for a set of graph patterns as observables, predict their number of occurrences over time. %
For example in models of biochemistry~\cite{danos2004computational,Harmer2010,danos2012graphs}, a pattern could represent a type of molecule, such that the number of occurrences of the pattern in a given graph that models the number of molecules of this type in a given cell. %
In the semantics of SRS based upon rule algebras~\cite{bp2018,nbSqPO2019,bdg2019,bp2020}, the evolution equations for such average expected pattern counts are derivable via so-called \emph{commutators}, which are computed by counting the distinct sequential compositions of rules representing \emph{transitions} and rules representing \emph{pattern observables}. %

In many such application scenarios graphs are  subject to structural constraints. This means that when composing rules, we have to ensure that any resulting composed rule respects these constraints, in the sense that an application of the composed rule reflects a sequential application of two given rules with a legal intermediate graph.
This causes the algorithmic challenge of avoiding where possible the construction of rule compositions forbidden by the constraints. %
Since a composition of two rules is defined based on a relation between the first rule's right-hand side and the second rule's left-hand side, we have to determine the set of all such relations that lead to composed rules respecting the structural constraints. 

In this paper, we develop a number of alternative strategies for addressing this  challenge, both theoretically and using the Z3 theorem prover with its Python interface. Specifically we tackle the computation of rule compositions in the Sesqui-Pushout (SqPO) approach~\cite{Corradini2006}, which provides the right level of generality for applications in biochemistry.

We first consider a \emph{direct} or \emph{``generate-and-test''} strategy, whereby in a ``generation''  phase all monic partial overlaps without taking into account the structural constraints are determined, and where in an ensuing ``testing'' phase each candidate overlap is scrutinised for whether or not its pushout object admits the embedding of any of the \emph{forbidden graph patterns} used to express the structural constraints.

As an alternative, a more modular strategy is based on a decomposition of the aforementioned forbidden graph patterns as  \emph{forbidden relation patterns}, which may then be utilised to design a search algorithm for monic partial overlaps that combines the search and the constraint-checking into a single operation, and thus in principle permits to avoid constructing a potentially enormous number of candidate overlaps not respecting the structural constraints. %
Mathematically, this second option is based upon an original contribution of this paper in terms of the theory of constraints: %
while the verification of a negative constraint is traditionally formulated in terms of a non-embedding condition of a \emph{forbidden pattern} into the target graph,  for the special case that the target graph is the pushout of a monic span one may equivalently reformulate this verification task as a non-embedding condition of a \emph{forbidden relation} into the span. %
We explicitly demonstrate the equivalence of the two approaches, i.e.\ given a graph pattern $P$ and all its relational decompositions $S$, for any composed graph $G$ over a span $s$ there exists an embedding of $P$ iff there exists a span in $S$ with a triple of compatible embeddings into $s$. 

Our third strategy for determining constraint-preserving rule overlaps is a variant of the ``generate-and-test'' strategy that is based upon a more \emph{implicit} encoding of the condition that the pushout of a given candidate rule overlap should respect the structural constraints. Approaches to the implementation of analysis techniques on graph transformation systems can be divided into \emph{native}, \emph{translation-based}, and \emph{hybrid}~\cite{DBLP:journals/corr/abs-1912-09607}. %
We follow a hybrid approach where all three strategies are implemented in Python using the Z3 SMT solver and theorem prover~\cite{deMoura2008} for computing decompositions and embeddings. %
Our approach was inspired by an active research area in graph rewriting that features a  number of other implementations of analysis techniques, e.g., for termination, confluence, or reachability as discussed in  \cite{DBLP:journals/corr/abs-1912-09607}. %
At the core of our implementation, we precisely mirror the category-theoretical structures encoding graphs, morphisms, overlaps, compositions of graphs and constraints in terms of a class architecture in Python and utilising the Z3 Python API.  %
The declarative nature of SMT solving, where models are only determined up to isomorphism, is a good fit for the categorical theory of graph rewriting, allowing a very direct translation. %
The aim of this implementation is to evaluate the relative performance in terms of time and space requirements of our three strategies. %
As a case study, we consider an example of SqPO rewriting of rigid multigraphs representing polymers in organic chemistry. %
A multigraph (which, in general, allows parallel edges) is rigid if it does not contain edges of the same type in parallel, starting or ending in the same node, or as parallel loops (see definition of patterns in Example~\ref{ex:rGraph}). %
Based on a simple set of rules for creating and deleting edges we can model polymer formation, such that the evolution equations derivable via the commutators computed could predict the number of polymers of a particular length and shape at any time of the reaction. %
The particular contribution of this paper is in developing the theory and experimental implementation of rule compositions with graph constraints following the three alternative strategies and a preliminary comparative evaluation of their scalability.

%%%%%%%%%%%%%%%%%%%%%%%%%%%%%%%%%%%%%%%%%%%%%%%%%%
\section{Rule Compositions for Conditional SqPO Rewriting}
%%%%%%%%%%%%%%%%%%%%%%%%%%%%%%%%%%%%%%%%%%%%%%%%%%

We provide some background on the recent extension of \emph{Sesqui-Pushout (SqPO) rewriting}~\cite{Corradini2006} to the setting of rules with application conditions~\cite{behrRaSiR} and describe the additional assumptions on the underlying category to admit a rule algebra construction~\cite{nbSqPO2019,bk2020a}. We refer the readers to~\cite{bk2020a} for the extended discussion of the mathematical concepts.

%%%%%%%%%%%%%%%%%%%%%%%%%%%%%%%%%%%%%%%%%%%%%%%%%%
\subsection{Categorical Setting}
%%%%%%%%%%%%%%%%%%%%%%%%%%%%%%%%%%%%%%%%%%%%%%%%%%

\begin{assumption}[\cite{behrRaSiR}]\label{as:main}
    We assume that $\bfC\equiv(\bfC,\cM)$ is a finitary $\cM$-adhesive category with $\cM$-effective unions, $\cM$-initial object, an epi-$\cM$-factorisation, existence of final pullback complements (FPCs) for all pairs of composable $\cM$-morphisms and with stability of $\cM$-morphisms under FPCs.
\end{assumption}

While a full discussion of these technical concepts is out of the scope of the present paper, we will comment at various points on some of the salient ideas and uses of the various assumptions made. Suffice it here to note that $\cM$ is a class of monomorphisms, and that we impose the finitarity-constraint for practical reasons, in that in our intended applications only \emph{finite} structures and their transformations are of interest, with a prominent example given as follows:

\begin{example}
    The prototypical category suited for rewriting in the above sense is the category $\mathbf{FinGraph}$ of \emph{finite directed multigraphs}. It is well-known~\cite{lack2005adhesive} that the category $\mathbf{Graph}$ of all (not necessarily finite) multigraphs is \emph{adhesive} (i.e.\ $\cM$-adhesive for $\cM=\mono{\mathbf{Graph}}$), and thus due to results of~\cite{gabriel2014} so is its finitary restriction $\mathbf{FinGraph}$. The finitary restriction preserves the epi-mono-factorisation, the property of mono-effective unions as well as the mono-initial object $\mIO$ (the empty graph). Finally, according to~\cite{Corradini2006}, FPCs exist for arbitrary pairs of composable monos, and monos are stable under FPCs.
\end{example}

%%%%%%%%%%%%%%%%%%%%%%%%%%%%%%%%%%%%%%%%%%%%%%%%%%
\subsection{Conditions in SqPO Rewriting}\label{sec:CSQPO}
%%%%%%%%%%%%%%%%%%%%%%%%%%%%%%%%%%%%%%%%%%%%%%%%%%

For the readers' convenience, we recall some of the relevant background material and standard notations (including a convenient shorthand notation for simple, i.e. non-nested conditions). Conceptually, application conditions are defined such as to constrain the \emph{matches} of rewriting rules. Another important special case are conditions over the $\cM$-initial object $\mIO$. Such \emph{global constraints} describe properties of all objects, such as invariants.

\begin{definition}[cf.\ e.g.\ \cite{habel2009correctness,ehrig2014mathcal}]
    Given an $\cM$-adhesive category $\bfC$ satisfying Assumption~\ref{as:main}, (nested) \emph{conditions} $\cond{\bfC}$ over $\bfC$ are recursively defined as follows:
    \begin{enumerate}
            \item For all objects $X\in \obj{\bfC}$, $\ac{true}_X$ is a condition.
            \item For every $\cM$-morphism $(f:X\hookrightarrow Y)\in\cM$ and for every condition $\ac{c}_Y\in\cond{\bfC}$ over $Y$, $\exists (f,\ac{c}_Y)$ is a condition.
        \item If $\ac{c}_X\in\cond{\bfC}$ is a condition over $X$, so is $\neg \ac{c}_X$.
        \item If $\ac{c}_X^{(1)},\ac{c}_X^{(2)}\in\cond{\bfC}$ are conditions over $X$, so is $\ac{c}_X^{(1)}\land \ac{c}_X^{(2)}$.
    \end{enumerate}
    The \emph{satisfaction} of a condition $\ac{c}_X$ by an $\cM$-morphism $(h:X\hookrightarrow Z)\in\cM$ is recursively defined as follows:
    \begin{enumerate}
        \item $h\vDash\ac{true}_X$.
        \item $h\vDash\exists(f:X\hookrightarrow Y,\ac{c}_Y)$ iff there exists $(g:Y\hookrightarrow Z)\in\cM$ such that $h=g\circ f$ and $g\vDash \ac{c}_Y$.
        \item $h\vDash\neg \ac{c}_X$ iff $h\neg\vDash\ac{c}_X$.
        \item $h\vDash \ac{c}_X^{(1)}\land \ac{c}_X^{(2)}$ iff $h\vDash \ac{c}_X^{(1)}$ and $h\vDash \ac{c}_X^{(2)}$.
    \end{enumerate}
Two conditions $\ac{c}_X,\ac{c}_X'\in\cond{\bfC}$ are \emph{equivalent}, denoted  $\ac{c}_X\equiv\ac{c}_X'$, iff for every $(h:X\hookrightarrow Z)\in\cM$ we find that $h\vDash\ac{c}_X$ implies $h\vDash\ac{c}_X'$ and vice versa.
\end{definition}
We will employ the following standard shorthand notations:
\begin{equation*}
    \exists(f:X\hookrightarrow Y):=\exists(f:X\hookrightarrow Y,\ac{true}_Y)\,,\quad
    \forall(f:X\hookrightarrow Y,\ac{c}_Y):=\neg \exists(f:X\hookrightarrow Y,\neg\ac{c}_Y)\,.
\end{equation*}

An auxiliary construction~\cite{habel2009correctness}  permits to extend conditions into a larger context.
\begin{theorem}\label{thm:Shift}
    With notations as above and for $(f:X\hookrightarrow Y)\in\cM$, there exists a \emph{shift construction} $\Shift$ such that
    \begin{equation}
    \begin{aligned}
        &\forall \ac{c}_X\in\cond{\bfC},\forall (h:X\hookrightarrow Z)\in\cM: \exists (g:Y\hookrightarrow Z)\in\cM: h=g\circ f\\
        &\Rightarrow \quad (h\vDash \ac{c}_X \Leftrightarrow g\vDash\Shift(f,\ac{c}_X))\,.
    \end{aligned}
    \end{equation}
\end{theorem}

For the computation of SqPO rule compositions for rules with conditions, we need an additional auxiliary construction for the calculus of conditions:
\begin{theorem}[\cite{behrRaSiR}, Thm.~7]\label{thm:Trans}
    Given a linear rule $r=(O\hookleftarrow K\hookrightarrow I)\in\Lin{\bfC}$ and a condition $\ac{c}_O\in\cond{\bfC}$ over $O$, there exists a \emph{transport construction} $\Trans$ such that for any object $X\in\obj{\bfC}$ and for any SqPO-admissible match $m\in\OMatch{r}{X}$ of $r$ into $X$, if $(m^{*}:O\hookrightarrow r_m(X))\in\cM$ denotes the comatch of $m$, the following holds:
\begin{equation}
    m^{*}\vDash \ac{c}_O \quad \Leftrightarrow \quad m\vDash\Trans(r,\ac{c}_O)\,.
\end{equation}
\end{theorem}

While an interesting set of theoretical results in its own right, the concrete implementation of the $\Shift$ and $\Trans$ constructions are not of relevance to the present paper, so we refer the interested readers to~\cite[Thms.~4 and~7]{behrRaSiR} for the precise details and further technical discussions.

%%%%%%%%%%%%%%%%%%%%%%%%%%%%%%%%%%%%%%%%%%%%%%%%%%
\subsection{Rigidity}
\label{ssec:rigidity}
%%%%%%%%%%%%%%%%%%%%%%%%%%%%%%%%%%%%%%%%%%%%%%%%%%

\begin{example}\label{ex:rGraph}
    Our running example throughout this paper will be the category $\mathbf{rGraph}$ of \emph{finite rigid directed multigraphs}. Referring to~\cite{Danos2014} for an extended discussion of the rigidity phenomenon, suffice it here to introduce this category as a refinement of the category $\mathbf{FinGraph}$ via imposing the following \emph{global constraint} formulated via a set $\cN$ of \emph{forbidden patterns}:\label{eq:gsc}
    \begin{equation}
        \ac{c}_{\cN}:=\bigwedge\limits_{N\in \cN} \neg\exists(\mIO\hookrightarrow N)\,,\quad 
        \cN:=\{%
        \inputtikz{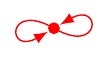},
        \inputtikzB{-0.15}{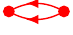},
         \inputtikzB{-0.15}{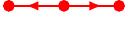},
          \inputtikzB{-0.15}{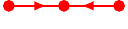}\}\,.
\end{equation}
Connected components of graphs in $\mathbf{rGraph}$ are thus either individual vertices, ``directed paths'' $\pi_n$ of edges (with lengths $n\geq1$) or closed ``directed loops'' $\lambda_n$ of edges (with lengths $n\geq2$),
\begin{equation*}
\includegraphics{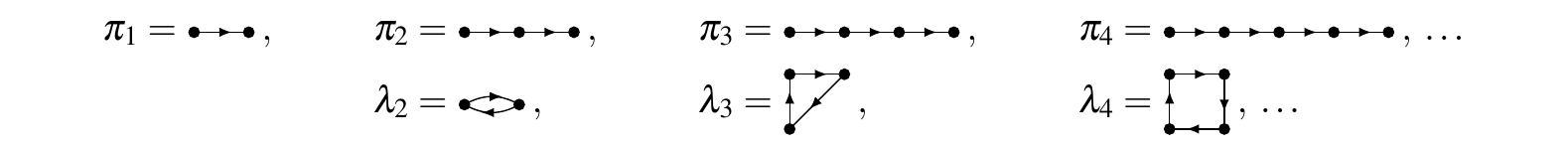}
\end{equation*}%
possibly in addition with individual loops on vertices, albeit we will only consider loop-less graphs in our applications. It is instructive to consider for this ``self-loop-less'' special case of graphs in $\mathbf{rGraph}$ the numbers of isomorphism classes as a function of the number of vertices. Following standard combinatorial arguments along the lines of~\cite{10.5555/1506267}, we can construct an ordinary generating function for these numbers by (1) dressing each vertex with a power of some formal variable $x$, each edge in a ``path'' with $y$ and each edge in a ``directed loop'' with $z$, and (2) multiplying the OGFs for copies of the different types of connected components:
\begin{equation*}
\begin{aligned}
\cG(x,y,z)&:=\sum_{k,\ell,m\geq 0} x^ky^{\ell}z^m\, \cdot(\text{\# of iso-classes with $k$ vertices and $\ell+m$ edges})\\
&= \tfrac{1}{1-x}
\prod_{p\geq 1}\tfrac{1}{1-x^{p+1}y^p}
\prod_{q\geq 2}\tfrac{1}{1-(xz)^q}\;
\xrightarrow{y,z\to1}
\tfrac{1-x}{((x;x)_{\infty})^2}\,.
\end{aligned}
\end{equation*}
The expansion of $\cG(x,1,1)$ (i.e.\ of the OGF of iso-classes counted by number of vertices) permits to identify this series as the OGF of the integer sequence \href{https://oeis.org/A000990}{A000990} on the OEIS database\footnote{%
N. J. A. Sloane, The On-Line Encyclopedia of Integer Sequences (\url{https://oeis.org})}, with sequence entries (for $n=0,\dotsc,10,\ldots$)
\begin{equation}
1, 1, 3, 5, 10, 16, 29, 45, 75, 115, 181, \ldots
\end{equation}
Comparing this to the sequence \href{https://oeis.org/A000273}{A000273} on the OEIS database, which gives the number of isomorphism classes of directed simple graphs without self-loops counted by number of vertices, and which reads for $n=0,\dotsc,10$
\begin{equation}
\begin{aligned}
1, 1, 3, 16, 218, 9608, 1540944, 882033440, 1793359192848, \\
13027956824399552, 341260431952972580352, \ldots\,,
\end{aligned}
\end{equation}
it transpires that the rigid directed graphs are enormously constrained in their structure compared to general simple directed graphs. This fact will in turn have a strong impact on the complexity of computing compositions of transformation rules based upon rigid graphs, in that any computation strategy based upon a ``generate-and-test'' approach will suffer from the enormous number of candidate rule overlaps. The way rigid graphs are restricted in their structure as compared to generic simple graphs is furthermore quite analogous to the way e.g.\ molecular graphs in the rewriting theory for organic chemistry~\cite{Benk2003,Andersen2016} or for biochemistry~\cite{danos2004formal,danos2012graphs,Boutillier:2018aa} are restricted as compared to generic (typed) simple graphs, which is why we take $\mathbf{rGraph}$ as a prototypical example of such types of applications of rewriting theory.
\end{example}

%%%%%%%%%%%%%%%%%%%%%%%%%%%%%%%%%%%%%%%%%%%%%%%%%%
\subsection{SqPO Direct Derivations and Rule Compositions}
%%%%%%%%%%%%%%%%%%%%%%%%%%%%%%%%%%%%%%%%%%%%%%%%%%

The theory of ``compositional'' SqPO rewriting as introduced in~\cite{nbSqPO2019,behrRaSiR,bk2020a} is an extension of the traditional SqPO theory~\cite{Corradini2006} by concurrency and associativity theorems that hold under suitable assumptions on the underlying categories. 

\begin{remark}
Contrary to standard conventions we will read rewriting rules from \textbf{input} to \textbf{output}, i.e.\ in particular from \textbf{right to left}. This is so as to be compatible with the standard mathematical convention of left-multiplication for matrices (see e.g.\ \cite{bdg2019} for an extended discussion).
\end{remark}

\begin{definition}\label{def:linRules}
Let $\bfC$ be an $\cM$-adhesive category satisfying Assumption~\ref{as:main}, and let $\LinAc{\bfC}$ denote the class of \emph{linear rewriting rules with (nested) conditions},
    \begin{equation}
        \LinAc{\bfC}:=\{(O\hookleftarrow K\hookrightarrow I;\ac{c}_I)\mid (K\hookrightarrow O),(K\hookrightarrow I)\in\cM,\ac{c}_I\in\cond{\bfC}\}\,.
    \end{equation}
    Let $\LinEq{\bfC}$ be the class of \emph{equivalence classes of linear rules with conditions} under the equivalence relation $\sim$ defined as follows:
    \begin{equation}
        (r,\ac{c}_I)\sim(r',\ac{c}_I'):\Leftrightarrow r\cong r' \land \ac{c}_I\equiv \ac{c}_I'\,.
    \end{equation}
    Here, $r\cong r'$ if and only if there exist isomorphisms $(\omega:O\rightarrow O')$, $(\kappa:K\rightarrow K')$ and $(\iota:I\rightarrow I')\in\iso{\bfC}$ such that the evident diagram commutes.
\end{definition}
The following definition provides a notion of direct derivations for SqPO rules with conditions.

\begin{definition}[\cite{behrRaSiR}, Def.~17; compare \cite{Corradini2006}, Def.~4]
\label{def:SqPOr}
Given an object $X\in \obj{\bfC}$ and a linear rule with condition $R=(r,\ac{c}_{I})\in \LinAc{\bfC}$, we define the \emph{set of admissible matches} $\OMatch{R}{X}$ as 
\begin{equation}
  \OMatch{R}{X}:=\{(m:I\rightarrow X)\in\cM\mid m\vDash \ac{c}_I\}\,.
\end{equation}

{
\makeatletter
\let\par\@@par
\par\parshape0
\everypar{}
\begin{wrapfigure}[5]{r}{0.35\linewidth}
\centering
\vspace{-2.45em}
\begin{equation}\label{eq:SqPO}\gdef\mycdScale{1}
\quad\inputtikz{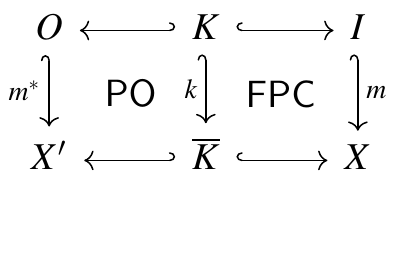}
\end{equation}
\end{wrapfigure}
\noindent A \emph{direct derivation of $X$ along $R$ with match $m\in\OMatch{R}{X}$} is defined via constructing the diagram on the right, with final pullback complement marked $\mathsf{FPC}$ and pushout marked $\mathsf{PO}$. We write $R_m(X):= X'$ for the object ``produced'' by the above diagram, and we refer to $m^{*}$ as the \emph{comatch of $m$}.
\par
\makeatother
}
\end{definition}
As is well-known in the graph-rewriting community, the notion of FPC along $\cM$-morphisms has a very natural interpretation in the setting of $\bfC$ being some form of graph-based category, e.g.\ $\bfC=\mathbf{FinGraph}$ (for which $\cM$ is the class of monomorphisms of finite graphs): considering a linear rule $(\mIO\hookleftarrow\mIO\hookrightarrow \bullet)$ encoding (read from right to left) the deletion of a vertex, choosing a match (i.e.\ an embedding) of the input vertex of the rule into some finite graph $X$ at a vertex with incident edges and forming the FPC effectively results in an intermediate graph $\overline{K}$ that is obtained from $X$ by removing the matched vertex and all incident edges. In contrast, a \emph{pushout complement (POC)} for the same rule and match to a vertex in $X$ with incident edges would not exist, highlighting a characteristic feature in \emph{Double-Pushout (DPO)} rewriting in that edges may not be ``implicitly'' deleted.

The second main definition of SqPO-type ``compositional'' rewriting theory is given by a notion of sequential composition for rules with application conditions.

\begin{definition}[\cite{behrRaSiR}]\label{def:SqPOcomp}
    With notations as above, let $R_j=(r_j,\ac{c}_{I_j})\in \LinEq{\bfC}$ be two equivalence classes of linear rules with conditions ($j=1,2$). Fixing representatives $(O_j\hookleftarrow K_j\hookrightarrow I_j;\ac{c}_{I_j})$ ($j=1,2$), a monic span $\mu=(I_2\hookleftarrow M_{21}\hookrightarrow O_1)$ is defined to be an \emph{SqPO-admissible match of $R_2$ into $R_1$}, denoted $\mu\in\RMatch{R_2}{R_1}$, if the pushout complement marked $\mathsf{POC}$ in the diagram below exists\footnote{Note that if the diagram is constructable, it follows from the properties of $\cM$-adhesive categories and from further properties stated in Assumption~\ref{as:main} that as indicated in the diagram all arrows are $\cM$-morphisms. More precisely, the technical properties utilised to prove this claim are the stability of $\cM$-morphisms under pushouts and pullbacks~\cite{ehrig2014mathcal} as well as the stability of $\cM$-morphisms under FPCs. The proof that this type of composition indeed leads to a suitable notion of concurrency theorem hinges also upon $\cM$-effective unions and the existence of an $\cM$-initial object~\cite{nbSqPO2019,behrRaSiR}. See Section~\ref{sec:CSQPO}, Theorems~\ref{thm:Shift}  and~\ref{thm:Trans} for notations and further details.},
\begin{equation}\label{eq:SqPOrComp}
    \inputtikz{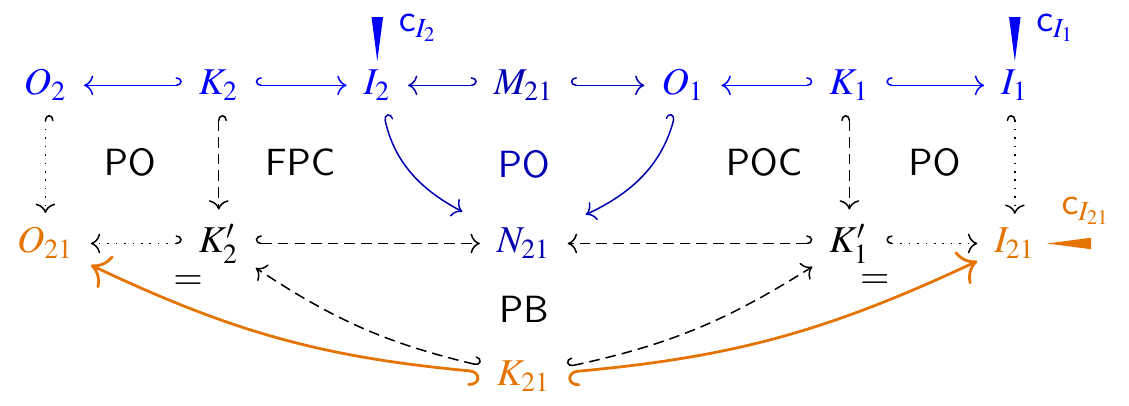}
\end{equation}
and if in addition $\ac{c}_{I_{21}}\not\equiv \ac{false}$, where
\begin{equation*}
\begin{aligned}
\ac{c}_{I_{21}}&=\Shift(I_1\hookrightarrow I_{21},\ac{c}_{I_{1}}) \land
\Trans(N_{21}\hookleftarrow K_1'\hookrightarrow I_{21},\Shift(I_2\hookrightarrow N_{21},\ac{c}_{I_2}))\,.
\end{aligned}
\end{equation*}
In this case, we define the \emph{SqPO-type composition of $R_2$ with $R_1$ along $\mu_{21}$} as the following equivalence class:
\begin{equation}
    \sqComp{R_2}{\mu}{R_1}:=[(O_{21}\hookleftarrow K_{21}\hookrightarrow I_{21};\ac{c}_{I_{21}})]_{\sim}\,.
\end{equation}
\end{definition}

\begin{example}[Ex.~\ref{ex:rGraph} continued]\label{ex:polymer}
    Seeing that rigid graphs as described in Example~\ref{ex:rGraph} may be considered as a toy model of \emph{polymers} that can form ``directed chains'' and ``directed loops'', one might ask how one could construct a stochastic dynamical model over $\mathbf{rGaph}$. Referring the interested readers to~\cite{nbSqPO2019,bdg2019,bk2020a} for a more in-depth discussion and the precise theoretical setup, for the present paper we focus on two types of contributions to a stochastic rewriting theory setup:
    \begin{enumerate}
        \item \textbf{Transitions} are modelled via rewriting rules with conditions $R_j=(r_j,\ac{c}_{I_j})\in \LinEq{\mathbf{rGraph}}$, one for each transition of the model. Choosing \emph{base rates} $\kappa_j\in \bR_{\geq 0}$ (i.e.\ rates of ``firings per second'') for each transition, the so-called ``mass-action semantics''~\cite{bk2020a} would render the likelihood of a given transition ``firing'' in an infinitesimal step proportional to its base rate times the number of matches of the rule into the current state.
        \item \textbf{Observables} implementing the detection and counting of \emph{patterns} are based upon rules of the special form~\cite{bk2020a}
        \begin{equation}
            R_{P,\ac{c}_P}=[(P\xhookleftarrow{id_P}P\xhookrightarrow{id_P}P; \ac{c}_P)]_{\sim}\,,
        \end{equation}
        where $P\in\obj{\bfC}_{\cong}$ is the ``pattern'', and where the condition $\ac{c}_P\in \cond{\bfC}$ in effect permits to implement pattern counting e.g. of the sort ``two vertices not linked by an edge''.
    \end{enumerate}
For our polymer example, one could start the construction of a stochastic model by considering a set of transitions that in effect are capable of rendering arbitrary finite rigid graphs through repeated ``firings'' of the transitions, i.e.\ the vertex deletion/creation and edge deletion/creation rules depicted in Figure~\ref{fig:ruleEx}\textbf{(a)}. Note that the application conditions necessary for these rules to respect the rigid graph constraints may be derived e.g.\ according to the algorithm presented in Corollary~\ref{cor:macs} (Section~\ref{sec:ccs}). In order to model certain physical effects such as saturation and steric hindrance effects, one might also consider to include more sophisticated transitions involving larger input and output graphs, such as the ones depicted in Figure~\ref{fig:ruleEx}\textbf{(b-d)}. Typical observables for the rule-sets presented would then e.g.\ be constructed in the form $R_{I_j,\ac{c}_{I_j}}$, i.e.\ from the input patterns and conditions (since these observables feature in the evolution equations of the respective stochastic rewriting systems, cf.\ \cite{bk2020a}).
\end{example}

\begin{figure}[t]
\label{fig:ruleEx}
\centering
\includegraphics{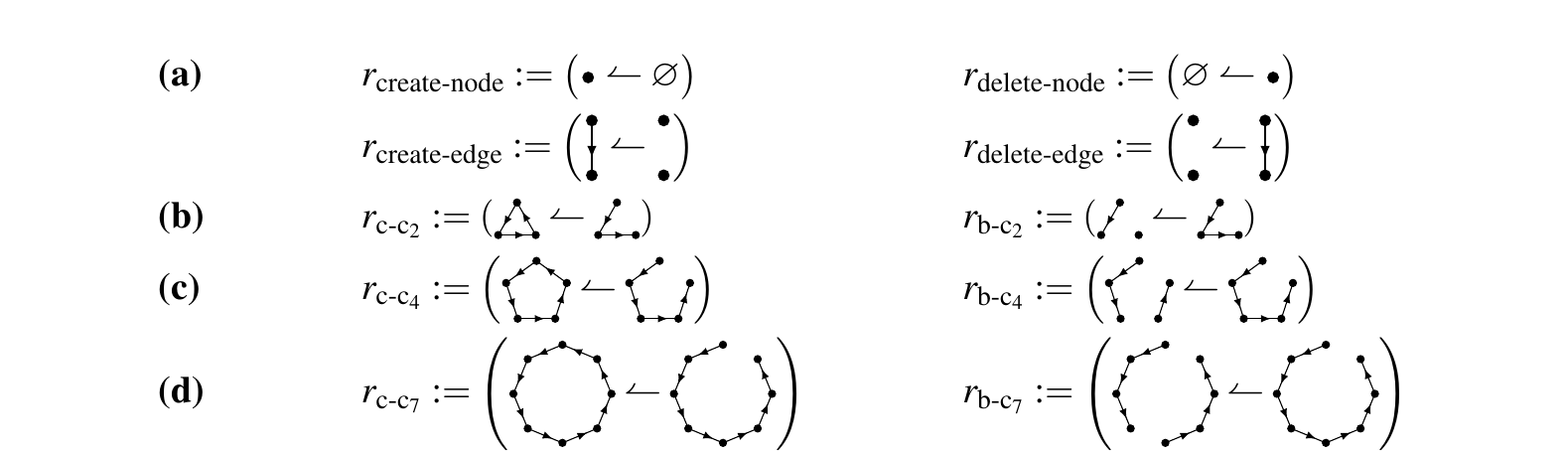}
\caption{Examples of transitions $R_j=(r_j,\ac{c}_{I_j})\in \LinEq{\mathbf{rGraph}}$ for rigid graph stochastic rewriting models: \textbf{(a)} the ``running example'' as a minimal setup, and \textbf{(b-d)} additional transitions involving larger graph patterns. Concretely, these rules model the two alternative options of closing a ``chain`` of edges into a loop (``\underline{c}reate-\underline{c}ycle``), or breaking up a chain of edges (``\underline{b}reak-\underline{c}hain``). The conditions $\ac{c}_{I_j}$ (not explicitly depicted) are defined as the respective constraint-preserving application conditions computed according to Corollary~\ref{cor:macs}.}
\end{figure}

%%%%%%%%%%%%%%%%%%%%%%%%%%%%%%%%%%%%%%%%%%%%%%%%%%
\section{Constraint-checking Strategies in Rule Compositions}\label{sec:ccs}
%%%%%%%%%%%%%%%%%%%%%%%%%%%%%%%%%%%%%%%%%%%%%%%%%%

When computing SqPO-type compositions of conditional rules, an algorithmically expensive step consists in verifying the satisfaction of both the global and application conditions in a given overlap, followed by computing the derived application conditions of the admissible composites. Both in order to experiment with the implementations of such algorithms via Z3 (see next Section) and out of a theoretical interest, we consider different implementation strategies for the steps involved in this rule composition operation. We will present here the ``direct'' strategy (i.e.\ following precisely the traditional constructions involving $\Shift$ and $\Trans$) as well as an alternative strategy based upon certain span (non-) embedding criteria. From this section onwards, we will fix a category $\bfC$ satisfying Assumption~\ref{as:main}, and assume a set of global conditions $\ac{c}_{\cN}$ as in~\eqref{eq:gsc} on objects. Let us further assume that all rules are endowed with application conditions that ensure the preservation of $\ac{c}_{\cN}$, and that the rules themselves are constructed of objects satisfying the constraint $\ac{c}_{\cN}$.

\begin{definition}[Direct Strategy] Given two linear rules with conditions $R_j\equiv(r_j,\ac{c}_{I_j})\in \LinAc{\bfC}$ ($j=1,2$) and a candidate match $\mu=(I_2\hookleftarrow M_{21}\hookrightarrow O_1)$, the \emph{Direct Strategy} to verify whether $\mu$ is an admissible match is defined as follows:
\begin{enumerate}
    \item Construct the pushout $(I_2\hookrightarrow N_{21}\hookleftarrow O_1)$ of $\mu$. Verify that $N_{21}\vDash\ac{c}_{\cN}$, and that $(I_2\hookrightarrow N_{21})\vDash\ac{c}_{I_2}$.
    \item If the pushout complement of $(N_{21}\hookleftarrow O_1,O_1\hookleftarrow K_1)$ exists, perform the SqPO-type composition according to Definition~\ref{def:SqPOcomp}, verifying that the application condition of the composite rule satisfies $\ac{c}_{I_{21}}\!\not\equiv\;\ac{false}$.
\end{enumerate} 
If both steps are successful, $\mu$ is an SqPO-admissible match of $R_2$ into $R_1$.
\end{definition}

As an alternative strategy, let us restate the ``forbidden patterns'' $N\in \cN$ via their \emph{pushout decompositions}, defined as follows:
\begin{definition}
    We define the \emph{set of ``forbidden relations''} $\cS_{\cN}$ as\footnote{As a further optimisation (tacitly employed from hereon and in our implementation), one may reduce the size of $\cS_{\cN}$ by only retaining one representative per isomorphism class of forbidden relations, with the same notion of isomorphism as the one utilised for linear rules (cf.\ Definition~\ref{def:linRules}).}
    \begin{equation}
        \cS_{\cN}:=\{
            s=(C_1\hookleftarrow D\hookrightarrow C_2)\mid C_1,D,C_2\vDash\ac{c}_{\cN}\land \exists N\in \cN:\; \pO{s}\cong N
        \}\,.
    \end{equation}
\end{definition}

\begin{example}
    For the category $\mathbf{rGraph}$ as introduced in Ex.~\ref{ex:rGraph}, one may compute the set of ``forbidden relations'' (with colours encoding the respective embeddings) as presented in Figure~\ref{fig:frs}.
\end{example}

\begin{figure}[t]
\label{fig:frs}
\centering
\includegraphics{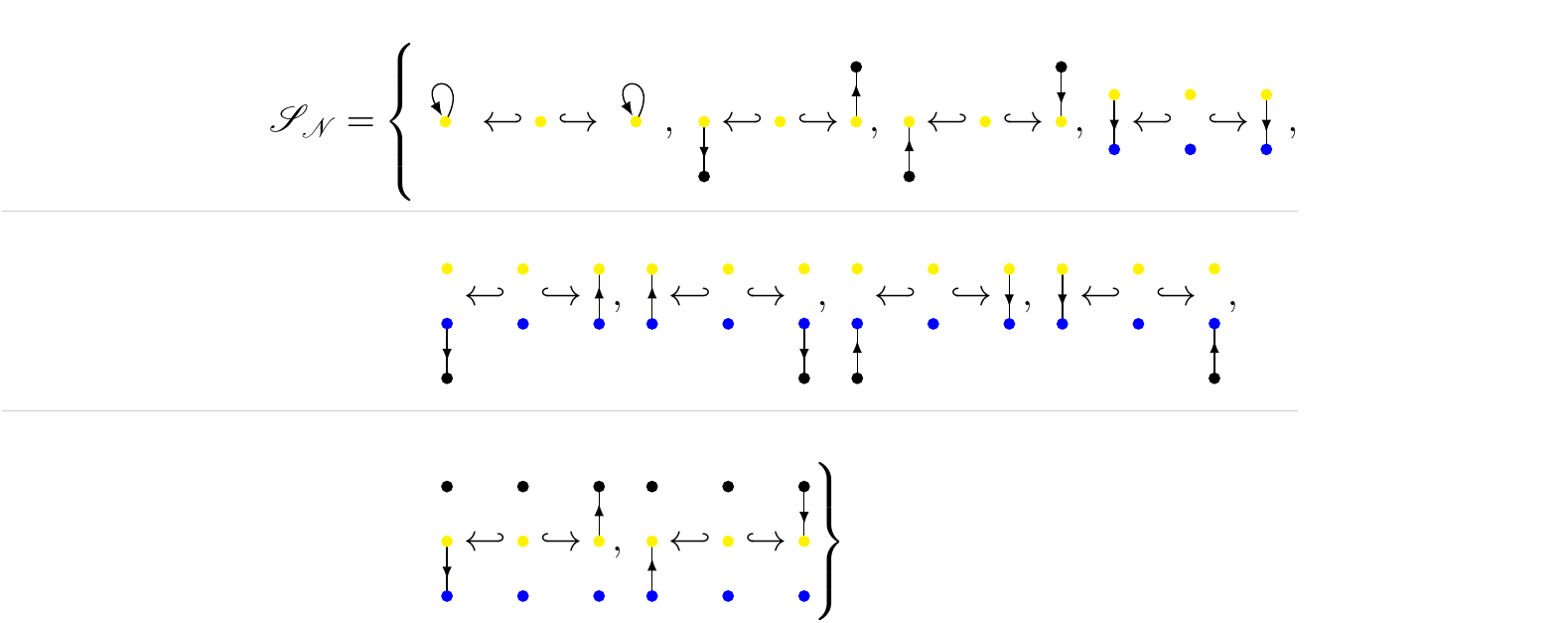}
\caption{Set of ``forbidden relations'' for the category $\mathbf{rGraph}$.}
\end{figure}

The pushout decompositions of forbidden patterns allow for a modular strategy to testing admissibility of rule overlaps that does not require to find embeddings of patterns into the pushout of the overlap, but  \emph{double-pullback embeddings (DPEs)} of forbidden spans into the monic spans representing the overlaps.
\begin{theorem}\label{thm:dpe}
    With notations and assumptions as above, given a pushout $P$ of a monic span $(I_1\hookleftarrow M_{21}\hookrightarrow O_1)$, the violation of $\ac{c}_{\cN}$ is equivalently verified as
    \begin{equation}\label{eq:dpeCond}
    \begin{aligned}
        P\not\vDash\ac{c}_{\cN}&\;\Leftrightarrow\;
        \exists s=(C_2\hookleftarrow D\hookrightarrow C_1)\in \cS_{\cN}:\\
        &\quad\qquad 
        \exists (C_2\hookrightarrow I_2),(D\hookrightarrow M_{21}),(C_1\hookrightarrow O_1)\in \mono{\bfC}:\\
        &\quad\qquad\quad (C_2\hookleftarrow D\hookrightarrow M_{21})=\pB{C_2\hookrightarrow I_2\hookleftarrow M_{21}}\\
        &\quad\qquad\land\; (C_2\hookleftarrow D\hookrightarrow M_{21})=\pB{C_2\hookrightarrow I_2\hookleftarrow M_{21}}\,.
    \end{aligned}
    \end{equation}
Here, each DPE mentioned in~\eqref{eq:dpeCond} encodes a commutative diagram of the form
\begin{equation}\label{eq:DPEdiag}
    \inputtikz{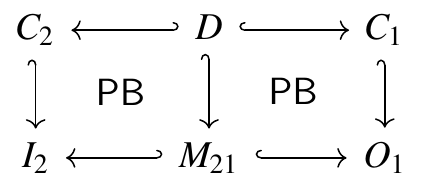}\,.
\end{equation} 
\end{theorem}
\begin{proof} The statement of the theorem follows from the $\cM$-vanKampen property of the category $\bfC$.\\

{\makeatletter
\let\par\@@par
\par\parshape0
\everypar{}
\begin{wrapfigure}[7]{l}{0.25\linewidth}
\centering
\vspace{-1em}
$\inputtikz{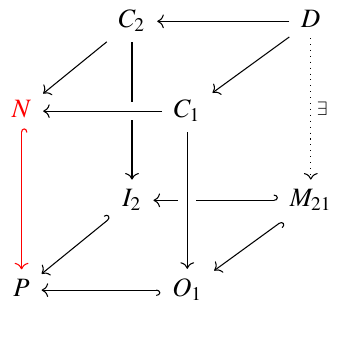}$
\end{wrapfigure}
\noindent For the $\Rightarrow$ direction, suppose that $P\not\vDash\ac{c}_{\cN}$, which entails that there exists an $N\in \cN$ and an embedding $(N\hookrightarrow P)$. Construct the commutative cube on the left via (1) taking pullbacks in order to obtain objects $C_2$ and $C_2$ and (2) letting $D$ be defined as the pullback of $(C_2\hookrightarrow N\hookleftarrow C_1)$. By stability of $\cM$ morphisms and by their decomposition properties, all arrows in the top square and all vertical arrows are $\cM$-morphisms. Since the bottom square is a pushout along $\cM$-morphisms and thus a pullback, and since the front and top squares are pullbacks, by pullback-pullback decomposition the back square is a pullback, and analogously so is the right square. Thus by the $\cM$-vanKampen property, the top square is a pushout, and we have proved that $(C_2\hookleftarrow D\hookrightarrow C_1)$ is in $\cS_{\cN}$.\par\makeatother}

For the $\Leftarrow$ direction, suppose the bottom pushout square as well as the back and left pullback squares were given (with all involved morphisms in $\cM$), and such that $(C_2\hookleftarrow D\hookrightarrow C_1)\in\cS_{\cN}$. Then letting $N$ be the pushout of $(C_2\hookleftarrow D\hookrightarrow C_1)$ (which by definition of $\cS_{\cN}$ entails that $N\in \cN$), there exists by universal property of the pushout a morphism $(N\hookrightarrow P)$. 
{\makeatletter
\let\par\@@par
\par\parshape0
\everypar{}
\begin{wrapfigure}[8]{l}{0.33\linewidth}
\centering
\vspace{-0.5em}
$\inputtikz{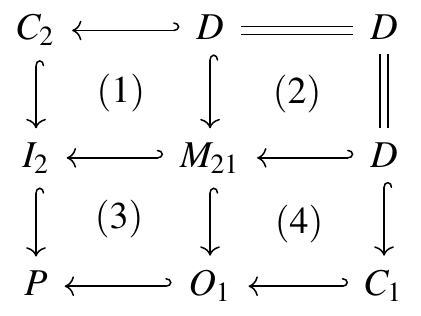}$
\end{wrapfigure}
\noindent It remains to demonstrate that $(N\rightarrow P)$ is in $\cM$. To this end, let us assemble the auxiliary on the left. By assumption, squares $(1)$ and $(4)$ are pullbacks, square $(3)$ a pushout along $\cM$-morphisms (and thus a pullback), while squares of the form $(2)$ are pullbacks by universal category theory. Consequently, we find by composition of pullbacks that $D$ is the pullback of $(C_2\hookrightarrow P \hookleftarrow C_1)$. Thus finally by virtue of the assumption of $\cM$-effective unions, we can confirm that $(N\hookrightarrow P)$ is in $\cM$.
\par
\makeatother}
\end{proof}

The alternative test for constraint satisfaction via DPEs of spans of $\cS_{\cN}$ according to Theorem~\ref{thm:dpe} permits to formulate the following alternative SqPO-type rule composition strategy:
\begin{definition}[DPE Strategy] Given two linear rules with conditions $R_j\equiv(r_j,\ac{c}_{I_j})\in \LinAc{\bfC}$ ($j=1,2$) and a candidate match $\mu=(I_2\hookleftarrow M_{21}\hookrightarrow O_1)$, the \emph{DPE Strategy} to verify whether $\mu$ is an admissible match is defined as follows:
\begin{enumerate}
    \item Verify that there does not exist any double-pullback embedding of a span of $\cS_{\cN}$ into the span $\mu$ (with DPEs as defined in~\eqref{eq:DPEdiag}).
    \item Verify that $(I_2\hookrightarrow N_{21})\vDash\ac{c}_{I_2}$.
    \item If the pushout complement of $(N_{21}\hookleftarrow O_1,O_1\hookleftarrow K_1)$ exists, perform the SqPO-type composition according to Definition~\ref{def:SqPOcomp}, verifying that the application condition $\ac{c}_{I_{21}}$ of the composite rule satisfies $\ac{c}_{I_{21}}\!\not{\!\equiv}\;\ac{false}$.
\end{enumerate} 
If all steps are successful, $\mu$ is an admissible SqPO-type match of $R_2$ into $R_1$.
\end{definition}

Finally, a useful corollary of Theorem~\ref{thm:dpe} is the following alternative algorithm for minimal constraint-preserving application conditions for SqPO-type rules (for the case as assumed in this section that rules themselves are constructed from objects satisfying the global constraint):
\begin{corollary}\label{cor:macs}
    With notations as above, given a ``plain'' linear rule $r=(O\hookleftarrow K\hookrightarrow I)\in \Lin{\bfC}$, perform the following steps:  for each $(C_2\hookleftarrow D\hookrightarrow C_1)\in\cS_{\cN}$,
    \begin{enumerate}
        \item find all \emph{pullback embeddings} of $(C_2\hookleftarrow D)$ into $(O\hookleftarrow K)$, i.e.\ pairs of $\cM$-morphisms $(C_2\hookrightarrow O)$ and $(D\hookrightarrow K)$ s.th.\ $D=\pB{C_2\hookleftarrow D\hookrightarrow K}$, and
        \item for each pullback embedding, construct $(C_1\hookrightarrow P\hookleftarrow I)$ by taking the pushout of the span $(C_1\hookleftarrow D\hookrightarrow K\hookrightarrow I)$; if $P\vDash\ac{c}_{\cN}$, then this pullback embedding contributes a negative condition of the form $\neg\exists(I\hookrightarrow P,\ac{true})$.
    \end{enumerate}
    The \emph{(minimal) constraint-preserving application condition} $\ac{c}_I$ is given by the conjunction over all individual contributions $\neg\exists(I\hookrightarrow P,\ac{true})$ computed above.  
\end{corollary}
\begin{proof} For the $\Rightarrow$ direction, let us assume that a given SqPO-type direct derivation along rule $(O\hookleftarrow K\hookrightarrow I)$ with candidate match $(m:I\hookrightarrow X_0)\in \cM$ results in an object $X_1$ with $X_1\not\vDash\ac{c}_{\cN}$. By definition of satisfiability and of $\ac{c}_{\cN}$, this entails that there exists an $N\in\cN$ and an $\cM$-morphism $(N\hookrightarrow X_1)\in\cM$. In complete analogy to the proof of Theorem~\ref{thm:dpe} as given in the previous section, construct the commutative cube below over the SqPO-type direct derivation diagram:
\begin{equation}
\gdef\mycdScale{0.7}
\inputtikz{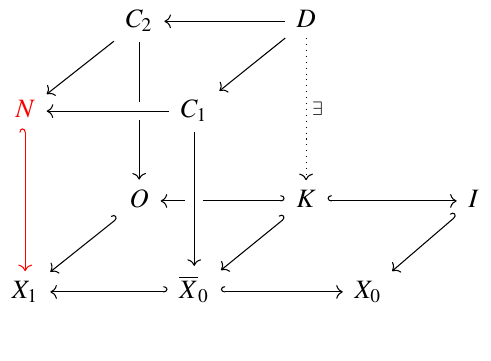}
\end{equation}
Here, $C_1$ and $C_2$ are constructed by taking pullbacks, which by stability of $\cM$-morphisms under pullbacks entails that the newly constructed morphism are also in $\cM$. $D$ is constructed as a pullback of $(C_2\hookrightarrow N\hookleftarrow C_1)$ (with $\cM$-morphisms in the resulting span), while the morphism $(D\rightarrow K)$ is an $\cM$-morphism by the decomposition property of $\cM$-morphisms. Since the bottom left square is a pushout along $\cM$-morphisms and thus a pullback, we find via pullback-pullback decomposition that also the squares $\cSq{D,C_2,O,K}$ (back left) and $\cSq{D,C_1,\overline{X}_0,K}$ (middle vertical) are pullbacks, and thus by the $\cM$-VK property the top left square is a pushout. Next, construct the following three pushouts:
\[
	{\color{h1color}\overline{D}}:=\pO{C_1\hookleftarrow D\hookrightarrow K}\,,\;
	{\color{h1color}\overline{C}_2}:=\pO{O\hookleftarrow K{\color{h1color}\hookrightarrow \overline{D}}}\,,\;
	{\color{h1color}P}:=\pO{{\color{h1color}\overline{D}\hookleftarrow} K\hookrightarrow I}\,.
\]
As depicted in the diagram below left, by the universal properties of the relevant pushouts and via $\cM$-effective unions, there exist morphisms (drawn below with dotted lines) that are in fact $\cM$-morphisms:
\begin{equation}\label{eq:appCProofD}
\inputtikz{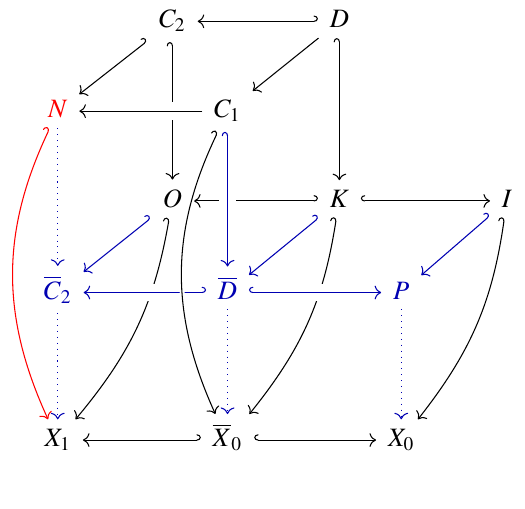}
\qquad\qquad 
\inputtikz{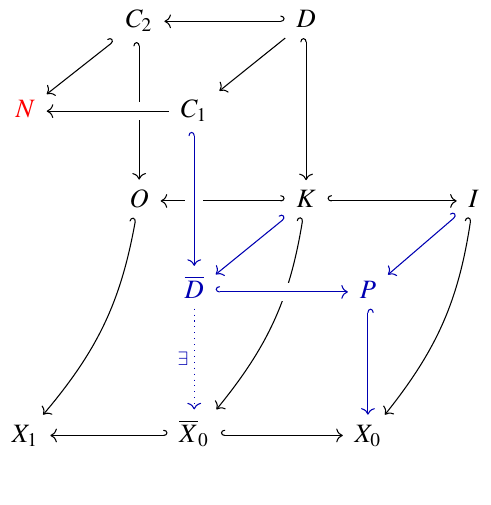}
\end{equation}
Finally, there are two cases to consider: if $P\vDash\ac{c}_{\cN}$, we have exhibited a $\cM$-morphism $(I\hookrightarrow P)$ through which $(I\hookrightarrow X_0)$ factors, and which thus by the above construction proves that the rewrite will lead to an $X_1$ with at least one embedding of a ``forbidden pattern'' $N\in \cN$. If on the other hand $P\not\vDash\ac{c}_{\cN}$, we have proved that $X_0\not\vDash \ac{c}_{\cN}$; consequently, the $\cM$-morphism $(I\hookrightarrow P)$ would in this case not contribute to a constraint-preserving application condition for $(O\hookleftarrow K\hookrightarrow I)$.

For the $\Leftarrow$ direction, let us assume we were given the data of the diagram below (i.e.\ a SqPO-type direct derivation of $X_0$ along candidate match $(I\hookrightarrow X_0)\in \cM$ as well as a pattern $P$ such that $(I\hookrightarrow X_0)$ factors through some $(I\hookrightarrow P)$ constructed according to the statement of the Corollary) as depicted in the diagram on the right of~\eqref{eq:appCProofD}. %
Here, in order to demonstrate the existence of the morphism ${\color{h1color}(\overline{D}\rightarrow \overline{X}_0)}$, let us first introduce a useful auxiliary formula: given a commutative diagram of $\cM$-morphisms as below left,
\begin{equation}
\inputtikz{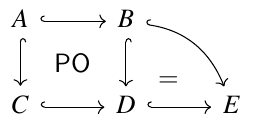}\rightsquigarrow 
\inputtikz{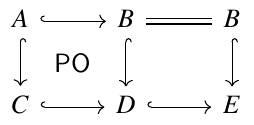}
\end{equation}
this data yields the diagram above right. Since the right square is a pullback and the left square a pushout along $\cM$-morphisms and thus also a pullback, we find by composition of pullbacks that $\cSq{A,C,E,B}$ is a pullback.

Back to the diagram on the right of~\eqref{eq:appCProofD}, since $\cSq{K,\overline{X}_0,X_0,I}$ is by assumption a final pullback complement and $\cSq{K,{\color{h1color}\overline{D}},X_0,I}$ is a pullback by virtue of the auxiliary formula, the universal property of FPCs entails the existence of $({\color{h1color}\overline{D}\rightarrow}\overline{X}_0)$, which by the decomposition property of $\cM$-morphisms is an $\cM$-morphism. Moreover, by FPC-pushout decomposition~\cite{behrRaSiR}, $\cSq{{\color{h1color}\overline{D}},\overline{X}_0,X_0,{\color{h1color}P}}$ is an FPC.

It then suffices to construct the pushout ${\color{h1color}\overline{C}_2}:=\pO{O\hookleftarrow K{\color{h1color}\hookrightarrow \overline{D}}}$, which through the universal properties of the various pushouts involved yields the existence of morphisms $(N{\color{h1color}\rightarrow \overline{C}})$ and $({\color{h1color}\overline{C}\rightarrow}X_1)$ (resulting in a diagram of the form in~\eqref{eq:appCProofD}. Utilising pushout-pushout decompositions (yielding that $\cSq{C_2,{\color{red}N},{\color{h1color}\overline{C}_2},O}$ and $\cSq{{\color{h1color}\overline{D}},{\color{h1color}\overline{C}_2},X_1,\overline{X}_0}$ are pushouts), stability of $\cM$-morphisms under pushouts, the above auxiliary formula (i.e.\ to demonstrate that $\cSq{C_2,{\color{red}N},X_1,O}$ and $\cSq{D,C_1,X_0,K}$ are pullbacks) and the property of $\cM$-effective unions, we find that $({\color{red}N}\rightarrow{\color{h1color}\overline{C}_2})\in \cM$ and $({\color{h1color}\overline{C}_2}\rightarrow X_1)\in \cM$, such that we have in summary exhibited a $\cM$-morphism $({\color{red}N}\hookrightarrow X_1)\in \cM$, which verifies that $X_1\not\vDash\ac{c}_{\cN}$.
\end{proof}

A heuristic solution for the problem of computing minimal constraint-preserving application conditions has been proposed in~\cite{NassarKAT19}. For the case of forbidden graph patterns $\cN$ and negative conditions $\ac{c}_{\cN}$, our result permits to avoid the use of the computationally expensive $\Shift$ and $\Trans$ constructions and subsequent minimisation \cite{habel2009correctness,ehrig2014mathcal}, yielding an equivalent solution for minimal $\ac{c}_{\cN}$-preserving constraints via directly taking advantage of the ``forbidden relations'' in $\cS_{\cN}$.

%%%%%%%%%%%%%%%%%%%%%%%%%%%%%%%%%%%%%%%%%%%%%%%%%%
\section{Implementation and Experimental Evaluation}\label{sec:Z3}
%%%%%%%%%%%%%%%%%%%%%%%%%%%%%%%%%%%%%%%%%%%%%%%%%%

We provide a Python package called  \texttt{ReSMT} (author: N. Behr) as an open-source project via GitLab\footnote{Official \texttt{ReSMT} repository: \url{https://gitlab.com/nicolasbehr/ReSMT}}, whose current version \texttt{0.0.3} contains both a full online documentation\footnote{Documentation: \url{https://nicolasbehr.gitlab.io/ReSMT/}} as well as dedicated \emph{GCM~2020 Supplementary Information} materials\footnote{GCM~2020 SI materials: \url{https://nicolasbehr.gitlab.io/ReSMT/py_and_ipynb_examples/GCM2020.html}}. The computational core of our implementation is a theory (a set of definitions and constraints) in the Z3 theorem prover\footnote{Z3 solver GitHub repository: \url{https://github.com/Z3Prover/z3}}~\cite{deMoura2008}, which due to its declarative nature allows for a direct encoding of our categorical constructions. Using Z3 through its Python API, we can reflect the structure of the theory in an object-oriented design, where classes represent concepts of the problem domain, such as graphs and graph morphisms, spans, rules, pushouts, etc., and generate assertions in Z3 encoding the defining properties of these categorical structures. However, a major challenge in implementing high-performance algorithms consists in evaluating and choosing between a number of alternative encodings of sets and set operations provided by the Z3 API. We first review these challenges in Section~\ref{sec:fs}, followed by a description of the key design choices taken for \texttt{ReSMT} in Section~\ref{sec:ros} and a presentation and evaluation of some experimental results in Section~\ref{sec:ee}. Section~\ref{sec:correct} comments on the correctness of the approach.

%%%%%%%%%%%%%%%%%%%%%%%%%%%%%%%%%%%%%%%%%%%%%%%%%%
\subsection{Key Challenge: Implementation of Finite Set Constructs}\label{sec:fs}
%%%%%%%%%%%%%%%%%%%%%%%%%%%%%%%%%%%%%%%%%%%%%%%%%%

Albeit finite sets are evidently a well-understood and fundamental theoretical concept, choosing an algorithmically advantageous representation of finite sets within Z3  is a non-trivial matter. As described in detail in the excellent Z3 tutorial~\cite{Bjrner2019}, the performance of Z3 when computing on sets of assertions  encoding a search over set structures hinges on the precise way in which the finiteness of the sets involved is communicated to the solver. In our prototypical early 2020 implementation, we had opted for what appeared to be a natural choice: a finite set such as $S=\{x_1,x_2,x_3\}$ may be encoded by declaring an abstract Z3 sort, instantiating Z3 constants of this sort for each element and asserting that any constant of these sorts must be equal to one of the explicitly instantiated constants:
\begin{minted}{python}
import z3				 # load the Z3 API package
S = z3.Solver()			   # instantiate a Z3 solver
X = z3.DeclareSort('X', S.ctx)	    # declare a sort in the context of S
x1,  x2,  x3 = z3.Consts('x1 x2 x3', X)   # instantiate one constant of sort X per set element
x = z3.Const('x', X)		      # instantiate an auxiliary constant
S.add(z3.ForAll(x, z3.Or([x == x1, x == x2, x == x3]))) # add the finiteness assertion to S
\end{minted}
Unfortunately, while of course logically correct, we found that this particular design choice results in poor performance of our algorithms for curating rule overlaps, both in execution time and memory consumption. As a remedy, our \texttt{ReSMT} package is based upon an alternative encoding of finite sets utilising the \emph{enumeration sorts}\footnote{See e.g. this tutorial: \url{https://ericpony.github.io/z3py-tutorial/advanced-examples.htm}} provided by Z3. For the aforementioned example, this alternative encoding reads as follows:
\begin{minted}[firstnumber=7]{python}
set, els = z3.EnumSort('set', ('x1', 'x2', 'x3'), S.ctx)
\end{minted}
The above code yields an instantiation of a Z3 enumeration sort \pythoninline{set} called \pythoninline{'set'} as well as a list of Z3 constants \pythoninline{els} (containing three constants of sort \pythoninline{set} with names \pythoninline{'x1'}, \pythoninline{'x2'} and \pythoninline{'x3'}). After the above call, the enumeration sort \pythoninline{set} as well as the individual constants may be utilised in a fully analogous fashion to other types of Z3 sorts and constants, yet internally provide additional information to the Z3 solver instance, encoding the finiteness axiom. We refer the  reader to a dedicated section\footnote{API experiments: \url{https://nicolasbehr.gitlab.io/ReSMT/py_and_ipynb_examples/ReSMT-API-experiments.html}} of the \texttt{ReSMT} online documentation for in-detail performance experiments and simplified versions of the  \texttt{ReSMT} algorithms for finding rule overlaps, demonstrating in particular the feasibility of using enumeration sorts for curating injective partial set overlaps with or without additional cardinality constraints.

%%%%%%%%%%%%%%%%%%%%%%%%%%%%%%%%%%%%%%%%%%%%%%%%%%
\subsection{Implementation Strategies for Curating Rule Overlaps}\label{sec:ros}
%%%%%%%%%%%%%%%%%%%%%%%%%%%%%%%%%%%%%%%%%%%%%%%%%%

Our current \texttt{ReSMT} implementation (version \texttt{0.0.3}) consists of approximately 1800 lines of \texttt{Python~3} source code and is provided with an in-depth API documentation\footnote{\texttt{ReSMT} API documentation: \url{https://nicolasbehr.gitlab.io/ReSMT/API-reference/Reference.html}}. The core task implemented in the form of three alternative variants of algorithms is the curation of rule overlaps modulo structural constraints, currently supported already for the general data structure of \emph{typed directed multigraphs}. We will present here some of the key aspects of our implementation in the setting of untyped directed multigraphs $G = (V_G, E_G, s_G:E_G\to V_G, t_G:E_G\to V_G)$ of relevance to the present paper. Assuming that the sets of vertices and edges of a given graph $G$ have been encoded as enumeration sorts \pythoninline{vG} and \pythoninline{eG} (cf. Section~\ref{sec:fs}), the source and target functions of $G$ are instantiated as Z3 functions:
\begin{minted}[firstnumber=8]{python}
srcG = z3.Function('srcG', eG, vG)
trgtG = z3.Function('trgtG', eG, vG)
\end{minted}
\emph{(Injective) graph homomorphisms} $\varphi\equiv(\varphi_V,\varphi_E):G_A\to G_B$  may then be encoded as follows:
\begin{minted}[firstnumber=10]{python}
phiV = z3.Function('phiV', vGA, vGB); phiE = z3.Function('phiE', eGA, eGB)
v1, v2 = z3.Consts('v1 v2', vGA); e, e1, e2 = z3.Consts('e e1 e2', eGA)
S.add(z3.ForAll(eA, z3.And(srcB(phiE(e)) == phiV(srcA(e)), trgtB(phiE(e)) == phiV(trgtA(e))))
S.add(z3.ForAll([v1,v2], z3.Implies(phiV(v1) == phiV(v2), v1 == v2))) # phiV injectivity
S.add(z3.ForAll([e1,e2], z3.Implies(phiE(e1) == phiE(e2), e1 == e2))) # phiE injectivity
\end{minted}
Finally, \emph{injective partial overlaps} of finite multigraphs may be efficiently implemented as \emph{bi-injective Boolean span predicates} $\Phi: G_A\times G_B\to \{\top, \bot\}$, which is particularly advantageous compared to a more direct encoding of monic spans of graphs in our exhaustive overlap-finding algorithms. The span predicate must satisfy the following logical assertions (i.e. bi-injectivity and a variant of a graph homomorphism property):
\begin{minted}[firstnumber=15]{python}
PhiV = z3.Function(vGA, vGB, z3.BoolSort); PhiE = z3.Function(eGA, eGB, z3.BoolSort)
vA, vA1, vA2 = z3.Consts('vA vA1 vA2', vGA); eA, eA1, eA2 = z3.Consts('eA eA1 eA2', eGA);
vB, vB1, vB2 = z3.Consts('vB vB1 vB2', vGB); eB, eB1, eB2 = z3.Consts('eB eB1 eB2', eGB);
astsBiInj = [z3.ForAll([vA1, vA2, vB], 
		       z3.Implies(z3.And(PhiV(vA1,vB)==True, PhiV(vA2,vB)==True), vA1==vA2)),
	     z3.ForAll([eA1, eA2, eB], 
		       z3.Implies(z3.And(PhiE(eA1,eB)==True, PhiE(eA2, eB)==True), eA1==eA2)),
	     z3.ForAll([vA, vB1, vB2], 
	               z3.Implies(z3.And(PhiV(vA,vB1)==True, PhiV(vA, vB2)==True), vB1==vB2)),
             z3.ForAll([eA,eB1,eB2], 
                       z3.Implies(z3.And(PhiE(eA,eB1)==True, PhiE(eA, eB2)==True), eB1==eB2))]
astsPhiHom = [z3.ForAll([eA, eB], 
		       z3.Implies(PhiE(eA,eB)==True, 
		                  z3.And(PhiV(srcA(eA),srcB(eB))==True, 
		               	  PhiV(trgtA(eA),trgtB(eB))==True)))
S.add(astsBiInj + astsPhiHom)
\end{minted}
In order to encode the necessary constructs for our three different strategies, we require a number of additional data structures and logical assertion generation methods, all of which are documented in full detail both in the \texttt{ReSMT} API reference as well as in several sets of tutorial examples on the documentation page of the package. For brevity, suffice it here to highlight the key conceptual difference between the \emph{direct} aka  ``generate-and-test'' strategy \ \pythoninline{resmt.datatypes.generateTDGoverlapsDirectstrategy} and its \emph{implicit} strategy variant %
\pythoninline{resmt.datatypes.generateTDGoverlapsImplicitstrategy}. The \emph{direct} strategy consists of a phase wherein all possible injective partial overlaps are determined via searching with Z3 for all possible monic Boolean span predicates as described above, followed by a second phase of forming a pushout object for each individual overlap and testing the pushouts for compliance with the structural constraints. In the \emph{implicit} strategy, both tasks are combined into a single set of logical assertions by instantiating \emph{ordinary sorts} (i.e.\ \emph{not} enumeration sorts) for the vertex and edge sets of the tentative pushout graphs $G_P$, and in addition asserting the existence of injective graph homomorphisms $\alpha:G_A\hookrightarrow G_P$ and $\beta:G_B\hookrightarrow G_P$ as well as the properties that $\alpha$ and $\beta$ should be \emph{jointly surjective} as well es compatible with the Boolean span predicate $\Phi$ (in the sense that $\Phi(a,b) = \top\Rightarrow \alpha(a)=\beta(b)$). What at first looks like a rather baroque tautology for the structure of the equivalent \emph{direct} strategy definition has in fact dramatic consequences for the performance of the overlap-finding algorithm, since the implicit encoding of the pushout object permits the simultaneous formulation of the forbidden pattern non-embedding constraints, such that in the \emph{implicit} strategy no non-consistent overlaps are ever constructed.

%%%%%%%%%%%%%%%%%%%%%%%%%%%%%%%%%%%%%%%%%%%%%%%%%%
\subsection{Experiments and Evaluation} \label{sec:ee}
%%%%%%%%%%%%%%%%%%%%%%%%%%%%%%%%%%%%%%%%%%%%%%%%%%

%%%%%%%%%%%%%%%%%%%%%%%%%%%%%%%%%%%%%%%%%%%%%%%%%%
{
\begin{table}
\label{tab:experiments}
\caption{Experimental results of the three different strategies for the examples {\bf P1} to {\bf P4}. For each experiment, the \emph{average total run-time} (over five complete runs) as well as the \emph{maximum RAM consumption} are provided. Here, a ``complete run'' consists in finding \emph{all} admissible overlaps. Note that for the case of the DPE strategy and example {\bf P3}, we report here on exception only on the result of a single run (due to the prohibitive run-time). All experiments were performed via our \texttt{ReSMT} Python package (version \texttt{0.0.3}) executed in \texttt{Python 3.8.2} and with \texttt{Z3-4.8.8.0} on a computer running \texttt{macOS Catalina~10.15.6} with a \emph{2.30GHz Intel(R) Core(TM) i7-3615QM CPU} and \emph{16 GB} of RAM. The entry {\color{red}n/a} signals experiments that were not possible to complete under the experimental conditions in an order of hours per run. Note that for experiment {\bf P4}, the number of candidate overlaps could only be estimated (cf. main text).}

\arrayrulecolor{black}
\footnotesize
\begin{tabular}{c*{4}{r@{\hskip 6pt}r}}
\toprule
{\bf Experiment name} & 
\multicolumn{2}{c}{{\bf P1}} & \multicolumn{2}{c}{{\bf P2}} & 
\multicolumn{2}{c}{{\bf P3}} & \multicolumn{2}{c}{{\bf P4}}\\ 
$(G_A, G_B)$ & 
\multicolumn{2}{c}{$(I_{\text{d-e}}, O_{\text{c-e}})$} & 
\multicolumn{2}{c}{$(I_{\text{b-c}_2}, O_{\text{c-c}_2})$} & 
\multicolumn{2}{c}{$(I_{\text{b-c}_4}, O_{\text{c-c}_4})$} & 
\multicolumn{2}{c}{$(I_{\text{b-c}_7}, O_{\text{c-c}_7})$} \\[0.5em]
\toprule
\# candidate overlaps & 
\multicolumn{2}{c}{8} & \multicolumn{2}{c}{49} & 
\multicolumn{2}{c}{2426} & 
\multicolumn{2}{c}{{\color{red}$1.17\cdot10^9\dotsc3.02\cdot10^{22}$}}\\
\# correct overlaps &
\multicolumn{2}{c}{5} & \multicolumn{2}{c}{4} & 
\multicolumn{2}{c}{6} & \multicolumn{2}{c}{9}\\
\toprule
\emph{direct} strategy & 
\SI{0.388}{\second} & (\SI{37.6}{\mega\byte}) & 
\SI{2.420}{\second} & (\SI{146.5}{\mega\byte}) & 
\SI{134.642}{\second} & (\SI{6466.2}{\mega\byte}) &
{\color{red}n/a} & ({\color{red}n/a})\\
\midrule
\emph{DPE} strategy & 
\SI{0.083}{\second} & (\SI{18.8}{\mega\byte}) & 
\SI{20.790}{\second} & (\SI{269.8}{\mega\byte}) & 
\SI{8205.913}{\second} & (\SI{5022.1}{\mega\byte}) &
{\color{red}n/a} & ({\color{red}n/a})\\
\midrule
\emph{implicit} strategy & 
\SI{0.093}{\second} & (\SI{16.8}{\mega\byte}) & 
\SI{0.242}{\second} & (\SI{19.5}{\mega\byte}) & 
\SI{1.079}{\second} & (\SI{30.1}{\mega\byte}) &
\SI{5.390}{\second} & (\SI{57.8}{\mega\byte})
\end{tabular}
\normalsize
\end{table}
}
%%%%%%%%%%%%%%%%%%%%%%%%%%%%%%%%%%%%%%%%%%%%%%%%%%

We present in Table~\ref{tab:experiments} the experimental results for the computation of the full set of admissible overlaps of rigid graphs for four different pairs of input graphs $(G_A,G_B)$ (chosen from the in- and output interfaces of the rules of Example~\ref{ex:polymer}), and for the three different strategies. We invite the readers to consult the \emph{GCM 2020 Supplementary Information} section of the \texttt{ReSMT} online documentation for the precise specification of each of the computational experiments. In particular, we provide precise instructions on how to replicate these experiments by executing \texttt{Jupyter} notebook \texttt{GCM2020.ipynb} provided as part of the \texttt{ReSMT} GitLab repository. The examples {\bf P1} to {\bf P4} were chosen to illustrate the fundamental effect that graph constraints may have on the computational complexity determining injective partial graph overlaps, and with the data type of rigid graphs as introduced in Example~\ref{ex:rGraph} providing a prototypical example of a graph-like structure defined via negative constraints. It may be easily verified that the pushout of an injective partial overlap of a ``chain'' of edges $\pi_n$ of length $n$ and a ``loop'' of edges $\lambda_{n+1}$ (with both notations as in Example~\ref{ex:rGraph}) can only be a rigid graph if either the overlap was empty, or if the ``chain'' fully embeds into the ``loop'' (in one of $n+1$ possible ways). This precisely explains the number of ``correct'' overlaps for examples {\bf P2} ($n=2$), {\bf P3} ($n=4$) and {\bf P4} ($n=7$), with a similar argument resulting in $5$ overlaps for example {\bf P1}. However, for the \emph{direct} strategy also the numbers of all injective partial overlaps \emph{without} taking into account the constraints is highly relevant, which as indicated in Table~\ref{tab:experiments} exhibits an extreme growth with growing size of the graphs. Notably, while for examples {\bf P1}, {\bf P2} and {\bf P3} it was possible to explicitly use our code for the case without constraints to experimentally find the precise numbers of constraint-less overlaps, in example {\bf P4} this number could only be very roughly estimated (with a lower and upper bound provided by the number of vertex-only overlaps and the number of overlaps of the joint sets of vertices and edges without taking into account the graph homomorphism structure, respectively). 

Inspecting the experimental results of Table~\ref{tab:experiments}, it becomes immediately evident from a comparison of run-times and maximal memory usages that, while the \emph{DPE} strategy poses a theoretical advantage over the \emph{direct} strategy in that it avoids computing the many overlap candidates inconsistent with the constraints (and for small examples such as {\bf P1} indeed shows slightly better performance), this theoretical advantage is in practice effectively voided by prohibitive memory consumption and a run-time of the order of several hours already for the relatively small example {\bf P3}. Both strategies are incapable of computing all admissible overlaps for example {\bf P4}, where the reason in the case of the \emph{direct} strategy is clearly the prohibitive number of candidate overlaps. On the other hand, the \emph{implicit} strategy appears to provide by far the most performant and robust implementation, with run-times and maximal memory consumptions for this strategy scaling rather favourably with the problem size. We thus posit that it is this type of strategy that will be capable of efficiently tackling the problem of computing rule compositions also in the highly relevant application fields of bio- and organo-chemistry, where one encounters graph constraints for the specification of the respective data types akin to the prototypical model of rigid graphs, and which we have thus identified as a viable avenue for future work and software developments.

%%%%%%%%%%%%%%%%%%%%%%%%%%%%%%%%%%%%%%%%%%%%%%%%%%
\subsection{Correctness} \label{sec:correct}
%%%%%%%%%%%%%%%%%%%%%%%%%%%%%%%%%%%%%%%%%%%%%%%%%%
We discuss the correctness of the strategies and their implementations. The correctness of the \emph{direct} and \emph{implicit} strategies is obvious from their definitions because they both encode the original statement of the problem of finding all spans between two given graphs whose pushout object does not contain any of the forbidden patters. The correctness of the \emph{DPE} strategy is established by Theorem~\ref{thm:dpe}.
The implementation of these strategies has not been formally verified, but their encoding by logical constraints reflects directly the definition of the categorical constructions involved. 

The implementation has been tested by a number of small examples where it is easy to calculate the expected results by hand. This provides validation that the basic constructions are implemented correctly. Where all three strategies have produced results we have compared them to establish that they are indeed functionally equivalent. 
In particular, when applied to the same example, they produced the same number of spans. 

%%%%%%%%%%%%%%%%%%%%%%%%%%%%%%%%%%%%%%%%%%%%%%%%%%
\section{Related Work}
%%%%%%%%%%%%%%%%

In~\cite{DBLP:journals/corr/abs-1912-09607} two of the authors contributed to a review of applications of SAT and SMT solvers in graph transformation. In summary, their use in analysing graph transformation systems has increased over the last years, e.g. for the computation of strongest post-conditions~\cite{DBLP:conf/gg/Corradini0N17,Corradini2019specifying}, the implementation of graph computation problems and algorithms~\cite{Kreowski2010,ermler2011graph}, model checking~\cite{IsenbergSW13}, the verification of graph properties (shapes) as invariants~\cite{Steenken2015, Steenken2011}, and termination analysis~\cite{Grez2015}. While these approaches share our general approach, they address very different analysis problems. i.e. none of the above compute rule overlaps. 

Apart from other references already discussed, there has been work on the derivation of differential equations from rewrite rules in the context of or inspired by \texttt{Kappa}~\cite{DanosFFHK08,DBLP:journals/eceasst/BapodraH10,danos2012graphs}. We share their motivations and aims in enabling a wealth of mathematical analysis for graph transformation systems. Such approaches are generalised by the rule algebra construction~\cite{bdg2016,bp2018,bdg2019,nbSqPO2019,bk2020a,bp2020}, which provides the context for our work. %
The computation of relations between rules subject to constraints is also an element of critical pair analysis~\cite{HeckelKT02}. Recent work in this area~\cite{DBLP:conf/icse/Lambers0TBH18} has focused on the level of granularity at which such relations are computed and presented. For example, in many applications it is not necessary to compute all relations, but instead to focus only on \emph{essential} ones with minimal conflicting overlap. 
When computing relations to derive commutators, the only meaningful abstraction is by isomorphism classes, but it may be interesting to explore the use of our implementation in critical pair and dependency analysis where often only selected relations are required.

%%%%%%%%%%%%%%%%%%%%%%%%%%%%%%%%%%%%%%%%%%%%%%%%%%
\section{Conclusion}
%%%%%%%%%%%%%%%%%%%%%%%%%%%%%%%%%%%%%%%%%%%%%%%%%%

We have presented results of both theoretical and experimental nature on the efficient construction of overlap spans with constraints for the computation of commutators for stochastic rewrite systems. On the theoretical side, we showed how for compositions of SqPO rules with application conditions, the constraint checking step in this construction can be modularised, avoiding the direct generate-and-test strategy (of first composing rules and then validating them) by replacing checks for forbidden graph patterns by checks for forbidden relation patterns, and how the same technique can be used to efficiently derive minimal negative conditions preserving forbidden graph constraints. On the experimental side, we have designed a prototypical implementation of our three alternative strategies in Z3 based on the categorical concepts and constructions of the theory. We evaluated the correctness of our algorithms by testing them against each other and compared their performance through a series of experiments, which permitted us to identify the \emph{implicit} strategy as the most robust and best performing one, which in contrast to the other two strategies scaled very well with increasing problem complexity.

We plan to implement the complete commutator computation, including constructing rules with conditions and counting their isomorphism classes, and to explore applications in rule-based models of biochemistry and adaptive networks to derive their evolution equations.

%%%%%%%%%%%%%%%%%%%%%%%%%%%%%%%%%%%%%%%%%%%%%%%%%
% BIBLIOGRAPHY

%%%%%%%%%%%%%%%%%%%%%%%%%%%%%%%%%%%%%%%%%%%%%%%%%


\begin{thebibliography}{41}
\providecommand{\bibitemdeclare}[2]{}
\providecommand{\surnamestart}{}
\providecommand{\surnameend}{}
\providecommand{\urlprefix}{Available at }
\providecommand{\url}[1]{\texttt{#1}}
\providecommand{\href}[2]{\texttt{#2}}
\providecommand{\urlalt}[2]{\href{#1}{#2}}
\providecommand{\doi}[1]{doi:\urlalt{http://dx.doi.org/#1}{#1}}
\providecommand{\bibinfo}[2]{#2}

\bibitemdeclare{incollection}{Andersen2016}
\bibitem{Andersen2016}
\bibinfo{author}{Jakob~L. \surnamestart Andersen\surnameend},
  \bibinfo{author}{Christoph \surnamestart Flamm\surnameend},
  \bibinfo{author}{Daniel \surnamestart Merkle\surnameend} \&
  \bibinfo{author}{Peter~F. \surnamestart Stadler\surnameend}
  (\bibinfo{year}{2016}): \emph{\bibinfo{title}{{A Software Package for
  Chemically Inspired Graph Transformation}}}.
\newblock In: {\sl \bibinfo{booktitle}{Graph Transformation (ICGT 2016)}}, {\sl
  \bibinfo{series}{LNCS}} \bibinfo{volume}{9761}, \bibinfo{publisher}{Springer
  International Publishing}, pp. \bibinfo{pages}{73--88},
  \doi{10.1007/978-3-319-40530-8_5}.

\bibitemdeclare{article}{DBLP:journals/eceasst/BapodraH10}
\bibitem{DBLP:journals/eceasst/BapodraH10}
\bibinfo{author}{Mayur \surnamestart Bapodra\surnameend} \&
  \bibinfo{author}{Reiko \surnamestart Heckel\surnameend}
  (\bibinfo{year}{2010}): \emph{\bibinfo{title}{From Graph Transformations to
  Differential Equations}}.
\newblock {\sl \bibinfo{journal}{Electron. Commun. Eur. Assoc. Softw. Sci.
  Technol.}} \bibinfo{volume}{30}, \doi{10.14279/tuj.eceasst.30.431}.

\bibitemdeclare{inproceedings}{nbSqPO2019}
\bibitem{nbSqPO2019}
\bibinfo{author}{Nicolas \surnamestart Behr\surnameend} (\bibinfo{year}{2019}):
  \emph{\bibinfo{title}{{Sesqui-Pushout Rewriting: Concurrency, Associativity
  and Rule Algebra Framework}}}.
\newblock In: {\sl \bibinfo{booktitle}{{Workshop on Graph Computation (GCM
  2019)}}}, {\sl \bibinfo{series}{EPTCS}} \bibinfo{volume}{309},
  \bibinfo{publisher}{Open Publishing Association}, pp.
  \bibinfo{pages}{23--52}, \doi{10.4204/eptcs.309.2}.

\bibitemdeclare{inproceedings}{bdg2016}
\bibitem{bdg2016}
\bibinfo{author}{Nicolas \surnamestart Behr\surnameend},
  \bibinfo{author}{Vincent \surnamestart Danos\surnameend} \&
  \bibinfo{author}{Ilias \surnamestart Garnier\surnameend}
  (\bibinfo{year}{2016}): \emph{\bibinfo{title}{Stochastic mechanics of graph
  rewriting}}.
\newblock In: {\sl \bibinfo{booktitle}{Proceedings of the 31st Annual
  {ACM}/{IEEE} Symposium on Logic in Computer Science ({LICS} 2016)}},
  \bibinfo{publisher}{{ACM} Press}, \doi{10.1145/2933575.2934537}.

\bibitemdeclare{article}{bdg2019}
\bibitem{bdg2019}
\bibinfo{author}{Nicolas \surnamestart Behr\surnameend},
  \bibinfo{author}{Vincent \surnamestart Danos\surnameend} \&
  \bibinfo{author}{Ilias \surnamestart Garnier\surnameend}
  (\bibinfo{year}{2020}): \emph{\bibinfo{title}{{Combinatorial Conversion and
  Moment Bisimulation for Stochastic Rewriting Systems}}}.
\newblock {\sl \bibinfo{journal}{{Logical Methods in Computer Science}}}
  \bibinfo{volume}{{Volume 16, Issue 3}}.
\newblock \urlprefix\url{https://lmcs.episciences.org/6628}.

\bibitemdeclare{unpublished}{behrRaSiR}
\bibitem{behrRaSiR}
\bibinfo{author}{Nicolas \surnamestart Behr\surnameend} \&
  \bibinfo{author}{Jean \surnamestart Krivine\surnameend}
  (\bibinfo{year}{2019}): \emph{\bibinfo{title}{{Compositionality of Rewriting
  Rules with Conditions}}}.
\newblock \urlprefix\href{https://arxiv.org/abs/1904.09322}{arXiv:1904.09322 [cs.LO]
}.

\bibitemdeclare{inproceedings}{bk2020a}
\bibitem{bk2020a}
\bibinfo{author}{Nicolas \surnamestart Behr\surnameend} \&
  \bibinfo{author}{Jean \surnamestart Krivine\surnameend}
  (\bibinfo{year}{2020}): \emph{\bibinfo{title}{{Rewriting theory for the life
  sciences: A unifying framework for CTMC semantics}}}.
\newblock In: {\sl \bibinfo{booktitle}{Graph Transformation (ICGT 2020)}}, {\sl
  \bibinfo{series}{LNCS}} \bibinfo{volume}{12150}, \bibinfo{publisher}{Springer
  International Publishing}, pp. \bibinfo{pages}{185--202},
  \doi{10.1007/978-3-030-51372-6_11}.

\bibitemdeclare{inproceedings}{bp2018}
\bibitem{bp2018}
\bibinfo{author}{Nicolas \surnamestart Behr\surnameend} \&
  \bibinfo{author}{Pawel \surnamestart Sobocinski\surnameend}
  (\bibinfo{year}{2018}): \emph{\bibinfo{title}{{Rule Algebras for Adhesive
  Categories}}}.
\newblock In: {\sl \bibinfo{booktitle}{27th EACSL Annual Conference on Computer
  Science Logic (CSL 2018)}}, {\sl \bibinfo{series}{LIPIcs}}
  \bibinfo{volume}{119}, \bibinfo{publisher}{Schloss Dagstuhl--Leibniz-Zentrum
  f{\"u}r Informatik}, pp. \bibinfo{pages}{11:1--11:21},
  \doi{10.4230/LIPIcs.CSL.2018.11}.

\bibitemdeclare{article}{bp2020}
\bibitem{bp2020}
\bibinfo{author}{Nicolas \surnamestart Behr\surnameend} \&
  \bibinfo{author}{Pawel \surnamestart Sobocinski\surnameend}
  (\bibinfo{year}{2020}): \emph{\bibinfo{title}{{Rule Algebras for Adhesive
  Categories}}}.
\newblock {\sl \bibinfo{journal}{{Logical Methods in Computer Science}}}
  \bibinfo{volume}{{Volume 16, Issue 3}}.
\newblock \urlprefix\url{https://lmcs.episciences.org/6615}.

\bibitemdeclare{article}{Benk2003}
\bibitem{Benk2003}
\bibinfo{author}{Gil \surnamestart Benk\"{o}\surnameend},
  \bibinfo{author}{Christoph \surnamestart Flamm\surnameend} \&
  \bibinfo{author}{Peter~F. \surnamestart Stadler\surnameend}
  (\bibinfo{year}{2003}): \emph{\bibinfo{title}{{A Graph-Based Toy Model of
  Chemistry}}}.
\newblock {\sl \bibinfo{journal}{J. Chem. Inf. Comput. Sci.}}
  \bibinfo{volume}{43}(\bibinfo{number}{4}), pp. \bibinfo{pages}{1085--1093},
  \doi{10.1021/ci0200570}.

\bibitemdeclare{incollection}{Bjrner2019}
\bibitem{Bjrner2019}
\bibinfo{author}{Nikolaj \surnamestart Bj{\o}rner\surnameend},
  \bibinfo{author}{Leonardo \surnamestart de~Moura\surnameend},
  \bibinfo{author}{Lev \surnamestart Nachmanson\surnameend} \&
  \bibinfo{author}{Christoph~M. \surnamestart Wintersteiger\surnameend}
  (\bibinfo{year}{2019}): \emph{\bibinfo{title}{Programming Z3}}.
\newblock In: {\sl \bibinfo{booktitle}{Engineering Trustworthy Software
  Systems}}, \bibinfo{publisher}{Springer International Publishing}, pp.
  \bibinfo{pages}{148--201}, \doi{10.1007/978-3-030-17601-3_4}.

\bibitemdeclare{article}{Boutillier:2018aa}
\bibitem{Boutillier:2018aa}
\bibinfo{author}{Pierre \surnamestart Boutillier\surnameend},
  \bibinfo{author}{Mutaamba \surnamestart Maasha\surnameend},
  \bibinfo{author}{Xing \surnamestart Li\surnameend},
  \bibinfo{author}{H{\'{e}}ctor~F \surnamestart Medina-Abarca\surnameend},
  \bibinfo{author}{Jean \surnamestart Krivine\surnameend},
  \bibinfo{author}{J{\'{e}}r{\^{o}}me \surnamestart Feret\surnameend},
  \bibinfo{author}{Ioana \surnamestart Cristescu\surnameend},
  \bibinfo{author}{Angus~G \surnamestart Forbes\surnameend} \&
  \bibinfo{author}{Walter \surnamestart Fontana\surnameend}
  (\bibinfo{year}{2018}): \emph{\bibinfo{title}{The Kappa platform for
  rule-based modeling}}.
\newblock {\sl \bibinfo{journal}{Bioinformatics}}
  \bibinfo{volume}{34}(\bibinfo{number}{13}), pp. \bibinfo{pages}{i583--i592},
  \doi{10.1093/bioinformatics/bty272}.

\bibitemdeclare{misc}{Grez2015}
\bibitem{Grez2015}
\bibinfo{author}{Harrie Jan~Sander \surnamestart Bruggink\surnameend}
  (\bibinfo{year}{2015}): \emph{\bibinfo{title}{Grez user manual}}.
\newblock \urlprefix\url{http://www.ti.inf.uni-due.de/research/tools/grez/}.

\bibitemdeclare{inproceedings}{Corradini2006}
\bibitem{Corradini2006}
\bibinfo{author}{Andrea \surnamestart Corradini\surnameend},
  \bibinfo{author}{Tobias \surnamestart Heindel\surnameend},
  \bibinfo{author}{Frank \surnamestart Hermann\surnameend} \&
  \bibinfo{author}{Barbara \surnamestart K{\"o}nig\surnameend}
  (\bibinfo{year}{2006}): \emph{\bibinfo{title}{{Sesqui-Pushout Rewriting}}}.
\newblock In: {\sl \bibinfo{booktitle}{Graph Transformations (ICGT 2006)}},
  {\sl \bibinfo{series}{LNCS}} \bibinfo{volume}{4178},
  \bibinfo{publisher}{Springer Berlin Heidelberg}, pp. \bibinfo{pages}{30--45},
  \doi{10.1007/11841883_4}.

\bibitemdeclare{inproceedings}{DBLP:conf/gg/Corradini0N17}
\bibitem{DBLP:conf/gg/Corradini0N17}
\bibinfo{author}{Andrea \surnamestart Corradini\surnameend},
  \bibinfo{author}{Barbara \surnamestart K{\"{o}}nig\surnameend} \&
  \bibinfo{author}{Dennis \surnamestart Nolte\surnameend}
  (\bibinfo{year}{2017}): \emph{\bibinfo{title}{Specifying Graph Languages with
  Type Graphs}}.
\newblock In: {\sl \bibinfo{booktitle}{Graph Transformation ({ICGT} 2017)}},
  {\sl \bibinfo{series}{LNCS}} \bibinfo{volume}{10373},
  \bibinfo{publisher}{Springer International Publishing}, pp.
  \bibinfo{pages}{73--89}, \doi{10.1007/978-3-319-61470-0\_5}.

\bibitemdeclare{article}{Corradini2019specifying}
\bibitem{Corradini2019specifying}
\bibinfo{author}{Andrea \surnamestart Corradini\surnameend},
  \bibinfo{author}{Barbara \surnamestart K{\"o}nig\surnameend} \&
  \bibinfo{author}{Dennis \surnamestart Nolte\surnameend}
  (\bibinfo{year}{2019}): \emph{\bibinfo{title}{Specifying graph languages with
  type graphs}}.
\newblock {\sl \bibinfo{journal}{J. Log. Algebr. Methods Program.}}
  \bibinfo{volume}{104}, pp. \bibinfo{pages}{176--200},
  \doi{10.1016/j.jlamp.2019.01.005}.

\bibitemdeclare{inproceedings}{danos2012graphs}
\bibitem{danos2012graphs}
\bibinfo{author}{Vincent \surnamestart Danos\surnameend},
  \bibinfo{author}{Jerome \surnamestart Feret\surnameend},
  \bibinfo{author}{Walter \surnamestart Fontana\surnameend},
  \bibinfo{author}{Russell \surnamestart Harmer\surnameend},
  \bibinfo{author}{Jonathan \surnamestart Hayman\surnameend},
  \bibinfo{author}{Jean \surnamestart Krivine\surnameend},
  \bibinfo{author}{Chris \surnamestart Thompson-Walsh\surnameend} \&
  \bibinfo{author}{Glynn \surnamestart Winskel\surnameend}
  (\bibinfo{year}{2012}): \emph{\bibinfo{title}{{Graphs, Rewriting and Pathway
  Reconstruction for Rule-Based Models}}}.
\newblock In: {\sl \bibinfo{booktitle}{Foundations of Software Technology and
  Theoretical Computer Science (FSTTCS 2012)}}, {\sl
  \bibinfo{series}{LIPIcs}}~\bibinfo{volume}{18}, \bibinfo{publisher}{Schloss
  Dagstuhl--Leibniz-Zentrum f{\"u}r Informatik}, pp. \bibinfo{pages}{276--288},
  \doi{10.4230/LIPIcs.FSTTCS.2012.276}.

\bibitemdeclare{inproceedings}{DanosFFHK08}
\bibitem{DanosFFHK08}
\bibinfo{author}{Vincent \surnamestart Danos\surnameend},
  \bibinfo{author}{J{\'{e}}r{\^{o}}me \surnamestart Feret\surnameend},
  \bibinfo{author}{Walter \surnamestart Fontana\surnameend},
  \bibinfo{author}{Russell \surnamestart Harmer\surnameend} \&
  \bibinfo{author}{Jean \surnamestart Krivine\surnameend}
  (\bibinfo{year}{2008}): \emph{\bibinfo{title}{{Rule-Based Modelling,
  Symmetries, Refinements}}}.
\newblock In: {\sl \bibinfo{booktitle}{Formal Methods in Systems Biology
  ({FMSB} 2008)}}, {\sl \bibinfo{series}{LNCS}} \bibinfo{volume}{5054},
  \bibinfo{publisher}{Springer}, pp. \bibinfo{pages}{103--122},
  \doi{10.1007/978-3-540-68413-8\_8}.

\bibitemdeclare{incollection}{Danos2014}
\bibitem{Danos2014}
\bibinfo{author}{Vincent \surnamestart Danos\surnameend},
  \bibinfo{author}{Reiko \surnamestart Heckel\surnameend} \&
  \bibinfo{author}{Pawel \surnamestart Sobocinski\surnameend}
  (\bibinfo{year}{2014}): \emph{\bibinfo{title}{{Transformation and Refinement
  of Rigid Structures}}}.
\newblock In: {\sl \bibinfo{booktitle}{Graph Transformation (ICGT 2014)}}, {\sl
  \bibinfo{series}{LNCS}} \bibinfo{volume}{8571}, \bibinfo{publisher}{Springer
  International Publishing}, pp. \bibinfo{pages}{146--160},
  \doi{10.1007/978-3-319-09108-2_10}.

\bibitemdeclare{article}{danos2004formal}
\bibitem{danos2004formal}
\bibinfo{author}{Vincent \surnamestart Danos\surnameend} \&
  \bibinfo{author}{Cosimo \surnamestart Laneve\surnameend}
  (\bibinfo{year}{2004}): \emph{\bibinfo{title}{Formal molecular biology}}.
\newblock {\sl \bibinfo{journal}{TCS}}
  \bibinfo{volume}{325}(\bibinfo{number}{1}), pp. \bibinfo{pages}{69--110},
  \doi{10.1016/j.tcs.2004.03.065}.

\bibitemdeclare{book}{danos2004computational}
\bibitem{danos2004computational}
\bibinfo{editor}{Vincent \surnamestart Danos\surnameend} \&
  \bibinfo{editor}{Vincent \surnamestart Schachter\surnameend}, editors
  (\bibinfo{year}{2004}): \emph{\bibinfo{title}{{Computational Methods in
  Systems Biology (CMSB 2004)}}}.
\newblock {\sl \bibinfo{series}{LNCS}} \bibinfo{volume}{3082},
  \bibinfo{publisher}{Springer Berlin Heidelberg}, \doi{10.1007/b107287}.

\bibitemdeclare{article}{DBLP:journals/fuin/EhrigGHLO12}
\bibitem{DBLP:journals/fuin/EhrigGHLO12}
\bibinfo{author}{Hartmut \surnamestart Ehrig\surnameend},
  \bibinfo{author}{Ulrike \surnamestart Golas\surnameend},
  \bibinfo{author}{Annegret \surnamestart Habel\surnameend},
  \bibinfo{author}{Leen \surnamestart Lambers\surnameend} \&
  \bibinfo{author}{Fernando \surnamestart Orejas\surnameend}
  (\bibinfo{year}{2012}): \emph{\bibinfo{title}{{$\mathcal{M}$}-Adhesive
  Transformation Systems with Nested Application Conditions. Part 2: Embedding,
  Critical Pairs and Local Confluence}}.
\newblock {\sl \bibinfo{journal}{Fundam. Inform.}}
  \bibinfo{volume}{118}(\bibinfo{number}{1-2}), pp. \bibinfo{pages}{35--63},
  \doi{10.3233/FI-2012-705}.

\bibitemdeclare{article}{ehrig2014mathcal}
\bibitem{ehrig2014mathcal}
\bibinfo{author}{Hartmut \surnamestart Ehrig\surnameend},
  \bibinfo{author}{Ulrike \surnamestart Golas\surnameend},
  \bibinfo{author}{Annegret \surnamestart Habel\surnameend},
  \bibinfo{author}{Leen \surnamestart Lambers\surnameend} \&
  \bibinfo{author}{Fernando \surnamestart Orejas\surnameend}
  (\bibinfo{year}{2014}): \emph{\bibinfo{title}{{$\mathcal{M}$-adhesive
  transformation systems with nested application conditions. Part 1:
  parallelism, concurrency and amalgamation}}}.
\newblock {\sl \bibinfo{journal}{Math. Struct. Comput. Sci.}}
  \bibinfo{volume}{24}(\bibinfo{number}{04}), \doi{10.1017/s0960129512000357}.

\bibitemdeclare{inproceedings}{ermler2011graph}
\bibitem{ermler2011graph}
\bibinfo{author}{Marcus \surnamestart Ermler\surnameend},
  \bibinfo{author}{Hans-J{\"o}rg \surnamestart Kreowski\surnameend},
  \bibinfo{author}{Sabine \surnamestart Kuske\surnameend} \&
  \bibinfo{author}{Caroline \surnamestart von Totth\surnameend}
  (\bibinfo{year}{2011}): \emph{\bibinfo{title}{{From Graph Transformation
  Units via MiniSat to GrGen.NET}}}.
\newblock In: {\sl \bibinfo{booktitle}{International Symposium on Applications
  of Graph Transformations with Industrial Relevance}}, {\sl
  \bibinfo{series}{LNCS}} \bibinfo{volume}{7233}, \bibinfo{publisher}{Springer,
  Berlin, Heidelberg}, pp. \bibinfo{pages}{153--168},
  \doi{10.1007/978-3-642-34176-2_14}.

\bibitemdeclare{book}{10.5555/1506267}
\bibitem{10.5555/1506267}
\bibinfo{author}{Philippe \surnamestart Flajolet\surnameend} \&
  \bibinfo{author}{Robert \surnamestart Sedgewick\surnameend}
  (\bibinfo{year}{2009}): \emph{\bibinfo{title}{Analytic Combinatorics}}.
\newblock \bibinfo{publisher}{Cambridge University Press},
  \doi{10.1017/CBO9780511801655}.

\bibitemdeclare{article}{gabriel2014}
\bibitem{gabriel2014}
\bibinfo{author}{Karsten \surnamestart Gabriel\surnameend},
  \bibinfo{author}{Benjamin \surnamestart Braatz\surnameend},
  \bibinfo{author}{Hartmut \surnamestart Ehrig\surnameend} \&
  \bibinfo{author}{Ulrike \surnamestart Golas\surnameend}
  (\bibinfo{year}{2014}): \emph{\bibinfo{title}{Finitary $\mathcal{M}$-adhesive
  categories}}.
\newblock {\sl \bibinfo{journal}{Math. Struct. Comput. Sci.}}
  \bibinfo{volume}{24}(\bibinfo{number}{04}), \doi{10.1017/S0960129512000321}.

\bibitemdeclare{article}{habel2009correctness}
\bibitem{habel2009correctness}
\bibinfo{author}{Annegret \surnamestart Habel\surnameend} \&
  \bibinfo{author}{Karl-Heinz \surnamestart Pennemann\surnameend}
  (\bibinfo{year}{2009}): \emph{\bibinfo{title}{Correctness of high-level
  transformation systems relative to nested conditions}}.
\newblock {\sl \bibinfo{journal}{Math. Struct. Comput. Sci.}}
  \bibinfo{volume}{19}(\bibinfo{number}{02}), pp. \bibinfo{pages}{245--296},
  \doi{10.1017/s0960129508007202}.

\bibitemdeclare{article}{Harmer2010}
\bibitem{Harmer2010}
\bibinfo{author}{Russ \surnamestart Harmer\surnameend},
  \bibinfo{author}{Vincent \surnamestart Danos\surnameend},
  \bibinfo{author}{J{\'{e}}r{\^{o}}me \surnamestart Feret\surnameend},
  \bibinfo{author}{Jean \surnamestart Krivine\surnameend} \&
  \bibinfo{author}{Walter \surnamestart Fontana\surnameend}
  (\bibinfo{year}{2010}): \emph{\bibinfo{title}{Intrinsic information carriers
  in combinatorial dynamical systems}}.
\newblock {\sl \bibinfo{journal}{Chaos: An Interdisciplinary Journal of
  Nonlinear Science}} \bibinfo{volume}{20}(\bibinfo{number}{3}), p.
  \bibinfo{pages}{037108}, \doi{10.1063/1.3491100}.

\bibitemdeclare{article}{DBLP:journals/mscs/HeckelCEL96}
\bibitem{DBLP:journals/mscs/HeckelCEL96}
\bibinfo{author}{Reiko \surnamestart Heckel\surnameend},
  \bibinfo{author}{Andrea \surnamestart Corradini\surnameend},
  \bibinfo{author}{Hartmut \surnamestart Ehrig\surnameend} \&
  \bibinfo{author}{Michael \surnamestart L{\"{o}}we\surnameend}
  (\bibinfo{year}{1996}): \emph{\bibinfo{title}{Horizontal and Vertical
  Structuring of Typed Graph Transformation Systems}}.
\newblock {\sl \bibinfo{journal}{Math. Struct. Comput. Sci.}}
  \bibinfo{volume}{6}(\bibinfo{number}{6}), pp. \bibinfo{pages}{613--648},
  \doi{10.1017/S0960129500070110}.

\bibitemdeclare{inproceedings}{HeckelKT02}
\bibitem{HeckelKT02}
\bibinfo{author}{Reiko \surnamestart Heckel\surnameend},
  \bibinfo{author}{Jochen~Malte \surnamestart K{\"{u}}ster\surnameend} \&
  \bibinfo{author}{Gabriele \surnamestart Taentzer\surnameend}
  (\bibinfo{year}{2002}): \emph{\bibinfo{title}{{Confluence of Typed Attributed
  Graph Transformation Systems}}}.
\newblock In: {\sl \bibinfo{booktitle}{Graph Transformation (ICGT 2002)}}, {\sl
  \bibinfo{series}{LNCS}} \bibinfo{volume}{2505},
  \bibinfo{publisher}{Springer}, pp. \bibinfo{pages}{161--176},
  \doi{10.1007/3-540-45832-8\_14}.

\bibitemdeclare{inproceedings}{DBLP:journals/corr/abs-1912-09607}
\bibitem{DBLP:journals/corr/abs-1912-09607}
\bibinfo{author}{Reiko \surnamestart Heckel\surnameend}, \bibinfo{author}{Leen
  \surnamestart Lambers\surnameend} \& \bibinfo{author}{Maryam~Ghaffari
  \surnamestart Saadat\surnameend} (\bibinfo{year}{2019}):
  \emph{\bibinfo{title}{{Analysis of Graph Transformation Systems: Native vs
  Translation-based Techniques}}}.
\newblock In: {\sl \bibinfo{booktitle}{{Workshop on Graph Computation (GCM
  2019)}}}, {\sl \bibinfo{series}{{EPTCS}}} \bibinfo{volume}{309},
  \bibinfo{publisher}{Open Publishing Association}, pp. \bibinfo{pages}{1--22},
  \doi{10.4204/EPTCS.309.1}.

\bibitemdeclare{inproceedings}{IsenbergSW13}
\bibitem{IsenbergSW13}
\bibinfo{author}{Tobias \surnamestart Isenberg\surnameend},
  \bibinfo{author}{Dominik \surnamestart Steenken\surnameend} \&
  \bibinfo{author}{Heike \surnamestart Wehrheim\surnameend}
  (\bibinfo{year}{2013}): \emph{\bibinfo{title}{Bounded Model Checking of Graph
  Transformation Systems via {SMT} Solving}}.
\newblock In: {\sl \bibinfo{booktitle}{Formal Techniques for Distributed
  Systems ({FMOODS/FORTE} 2013)}}, {\sl \bibinfo{series}{LNCS}}
  \bibinfo{volume}{7892}, \bibinfo{publisher}{Springer, Berlin, Heidelberg},
  pp. \bibinfo{pages}{178--192}, \doi{10.1007/978-3-642-38592-6\_13}.

\bibitemdeclare{inproceedings}{Kreowski2010}
\bibitem{Kreowski2010}
\bibinfo{author}{Hans-J{\"o}rg \surnamestart Kreowski\surnameend},
  \bibinfo{author}{Sabine \surnamestart Kuske\surnameend} \&
  \bibinfo{author}{Robert \surnamestart Wille\surnameend}
  (\bibinfo{year}{2010}): \emph{\bibinfo{title}{Graph transformation units
  guided by a SAT solver}}.
\newblock In: {\sl \bibinfo{booktitle}{Graph Transformations (ICGT 2010)}},
  {\sl \bibinfo{series}{LNCS}} \bibinfo{volume}{6372},
  \bibinfo{organization}{Springer, Berlin, Heidelberg}, pp.
  \bibinfo{pages}{27--42}, \doi{10.1007/978-3-642-15928-2_3}.

\bibitemdeclare{article}{lack2005adhesive}
\bibitem{lack2005adhesive}
\bibinfo{author}{Stephen \surnamestart Lack\surnameend} \&
  \bibinfo{author}{Pawe{\l} \surnamestart Soboci{\'{n}}ski\surnameend}
  (\bibinfo{year}{2005}): \emph{\bibinfo{title}{{Adhesive and quasiadhesive
  categories}}}.
\newblock {\sl \bibinfo{journal}{{RAIRO} - Theoretical Informatics and
  Applications}} \bibinfo{volume}{39}(\bibinfo{number}{3}), pp.
  \bibinfo{pages}{511--545}, \doi{10.1051/ita:2005028}.

\bibitemdeclare{inproceedings}{DBLP:conf/icse/Lambers0TBH18}
\bibitem{DBLP:conf/icse/Lambers0TBH18}
\bibinfo{author}{Leen \surnamestart Lambers\surnameend},
  \bibinfo{author}{Daniel \surnamestart Str{\"{u}}ber\surnameend},
  \bibinfo{author}{Gabriele \surnamestart Taentzer\surnameend},
  \bibinfo{author}{Kristopher \surnamestart Born\surnameend} \&
  \bibinfo{author}{Jevgenij \surnamestart Huebert\surnameend}
  (\bibinfo{year}{2018}): \emph{\bibinfo{title}{Multi-granular conflict and
  dependency analysis in software engineering based on graph transformation}}.
\newblock In: {\sl \bibinfo{booktitle}{International Conference on Software
  Engineering ({ICSE} 2018)}}, \bibinfo{publisher}{{ACM}}, pp.
  \bibinfo{pages}{716--727}, \doi{10.1145/3180155.3180258}.

\bibitemdeclare{article}{DBLP:journals/entcs/LevendovszkyPE07}
\bibitem{DBLP:journals/entcs/LevendovszkyPE07}
\bibinfo{author}{Tihamer \surnamestart Levendovszky\surnameend},
  \bibinfo{author}{Ulrike \surnamestart Prange\surnameend} \&
  \bibinfo{author}{Hartmut \surnamestart Ehrig\surnameend}
  (\bibinfo{year}{2007}): \emph{\bibinfo{title}{Termination Criteria for {DPO}
  Transformations with Injective Matches}}.
\newblock {\sl \bibinfo{journal}{Electron. Notes Theor. Comput. Sci.}}
  \bibinfo{volume}{175}(\bibinfo{number}{4}), pp. \bibinfo{pages}{87--100},
  \doi{10.1016/j.entcs.2007.04.019}.

\bibitemdeclare{incollection}{deMoura2008}
\bibitem{deMoura2008}
\bibinfo{author}{Leonardo \surnamestart de~Moura\surnameend} \&
  \bibinfo{author}{Nikolaj \surnamestart Bj{\o}rner\surnameend}
  (\bibinfo{year}{2008}): \emph{\bibinfo{title}{{Z3: An Efficient {SMT}
  Solver}}}.
\newblock In: {\sl \bibinfo{booktitle}{Tools and Algorithms for the
  Construction and Analysis of Systems}}, \bibinfo{publisher}{Springer Berlin
  Heidelberg}, pp. \bibinfo{pages}{337--340},
  \doi{10.1007/978-3-540-78800-3_24}.

\bibitemdeclare{inproceedings}{NassarKAT19}
\bibitem{NassarKAT19}
\bibinfo{author}{Nebras \surnamestart Nassar\surnameend}, \bibinfo{author}{Jens
  \surnamestart Kosiol\surnameend}, \bibinfo{author}{Thorsten \surnamestart
  Arendt\surnameend} \& \bibinfo{author}{Gabriele \surnamestart
  Taentzer\surnameend} (\bibinfo{year}{2019}):
  \emph{\bibinfo{title}{{Constructing Optimized Validity-Preserving Application
  Conditions for Graph Transformation Rules}}}.
\newblock In: {\sl \bibinfo{booktitle}{Graph Transformation ({ICGT} 2019)}},
  {\sl \bibinfo{series}{LNCS}} \bibinfo{volume}{11629},
  \bibinfo{publisher}{Springer}, pp. \bibinfo{pages}{177--194},
  \doi{10.1007/978-3-030-23611-3\_11}.

\bibitemdeclare{inproceedings}{DBLP:conf/gg/RangelLKEB08}
\bibitem{DBLP:conf/gg/RangelLKEB08}
\bibinfo{author}{Guilherme \surnamestart Rangel\surnameend},
  \bibinfo{author}{Leen \surnamestart Lambers\surnameend},
  \bibinfo{author}{Barbara \surnamestart K{\"{o}}nig\surnameend},
  \bibinfo{author}{Hartmut \surnamestart Ehrig\surnameend} \&
  \bibinfo{author}{Paolo \surnamestart Baldan\surnameend}
  (\bibinfo{year}{2008}): \emph{\bibinfo{title}{{Behavior Preservation in Model
  Refactoring Using {DPO} Transformations with Borrowed Contexts}}}.
\newblock In: {\sl \bibinfo{booktitle}{Graph Transformations ({ICGT} 2008)}},
  {\sl \bibinfo{series}{LNCS}} \bibinfo{volume}{5214},
  \bibinfo{publisher}{Springer}, pp. \bibinfo{pages}{242--256},
  \doi{10.1007/978-3-540-87405-8\_17}.

\bibitemdeclare{phdthesis}{Steenken2015}
\bibitem{Steenken2015}
\bibinfo{author}{Dominik \surnamestart Steenken\surnameend}
  (\bibinfo{year}{2015}): \emph{\bibinfo{title}{Verification of infinite-state
  graph transformation systems via abstraction.}}
\newblock Ph.D. thesis, \bibinfo{school}{University of Paderborn}.
\newblock
  \urlprefix\url{https://digital.ub.uni-paderborn.de/ubpb/urn/urn:nbn:de:hbz:466:2-15768}.

\bibitemdeclare{inproceedings}{Steenken2011}
\bibitem{Steenken2011}
\bibinfo{author}{Dominik \surnamestart Steenken\surnameend},
  \bibinfo{author}{Heike \surnamestart Wehrheim\surnameend} \&
  \bibinfo{author}{Daniel \surnamestart Wonisch\surnameend}
  (\bibinfo{year}{2011}): \emph{\bibinfo{title}{Sound and complete abstract
  graph transformation}}.
\newblock In: {\sl \bibinfo{booktitle}{Brazilian Symposium on Formal Methods}},
  \bibinfo{organization}{Springer}, pp. \bibinfo{pages}{92--107},
  \doi{10.1007/978-3-642-25032-3_7}.

\end{thebibliography}
\end{document}